\documentclass[a4paper,UKenglish,cleveref, autoref, thm-restate]{lipics-v2021}
\usepackage[]{algorithm2e}

\newcommand{\enquote}[1]{``#1''}
\newcommand{\abs}[1]{\left|#1\right|}
\newcommand{\eps}{\varepsilon}

\newcommand\floor[1]{\left\lfloor #1 \right\rfloor}
\newcommand{\mysqrt}[2]{\sqrt[\leftroot{-1}\uproot{4}\scriptstyle #1]{#2}}
\newcommand{\reals}{\mathbb{R}}

\newcommand{\naturals}{\mathbb{N}}

\newcommand{\DSLSH}{\textit{\mbox{Discrete-sample-LSH}}\xspace}
\newcommand{\dslsh}{\textit{\mbox{discrete-sample-LSH}}\xspace}
\newcommand{\RPLSH}{\textit{\mbox{Random-point-LSH}}\xspace}
\newcommand{\rplsh}{\textit{\mbox{random-point-LSH}}\xspace}

\newcommand{\mrlsh}{\textit{\mbox{mean-reduce-LSH}}\xspace}
\newcommand{\VALSH}{\textit{\mbox{Vertical-alignment-LSH}}\xspace}
\newcommand{\valsh}{\textit{\mbox{vertical-alignment-LSH}}\xspace}
\newcommand{\SCLSH}{\textit{\mbox{Slide-clone-LSH}}\xspace}
\newcommand{\sclsh}{\textit{\mbox{slide-clone-LSH}}\xspace}
\newcommand{\SSLSH}{\textit{\mbox{Step-shift-LSH}}\xspace}
\newcommand{\sslsh}{\textit{\mbox{step-shift-LSH}}\xspace}
\newcommand{\LSH}{\mbox{\small LSH}\xspace}

\title{Locality Sensitive Hashing for Efficient Similar Polygon Retrieval}

\author{Haim Kaplan}{School of Computer Science, Tel Aviv University, Tel Aviv.}{haimk@tau.ac.il}{}{}
\author{Jay Tenenbaum}{School of Computer Science, Tel Aviv University, Tel Aviv.}{jayktenenbaum@gmail.com}{}{}
\authorrunning{H. Kaplan and J. Tenenbaum}
\Copyright{H. Kaplan and J. Tenenbaum}

\ccsdesc[300]{Theory of computation~Data structures design and analysis}
\ccsdesc[500]{Theory of computation~Computational geometry}
\ccsdesc[500]{Information systems~Information retrieval}
\keywords{Locality sensitive hashing, polygons, turning function, \texorpdfstring{$ L_p $}{} distance, nearest neighbors, similarity search}

\nolinenumbers 
\hideLIPIcs  

\funding{This work was supported by ISF grant no.\ 1595/19, GIF grant no.\ 1367/2017 and the Blavatnik Foundation.}

\EventEditors{Markus Bl\"{a}ser and Benjamin Monmege}
\EventNoEds{2}
\EventLongTitle{38th International Symposium on Theoretical Aspects of Computer Science (STACS 2021)}
\EventShortTitle{STACS 2021}
\EventAcronym{STACS}
\EventYear{2021}
\EventDate{March 16--19, 2021}
\EventLocation{Saarbr\"{u}cken, Germany (Virtual Conference)}
\EventLogo{}
\SeriesVolume{187}
\ArticleNo{40}

\begin{document}
	\maketitle
	\begin{abstract}
Locality Sensitive Hashing (LSH) is an effective method of indexing a set of items to support efficient nearest neighbors queries in high-dimensional spaces.
The basic idea of LSH is that similar items should produce hash collisions with higher probability than dissimilar items.

We study LSH for (not necessarily convex) polygons, and use it to give efficient data structures for similar shape retrieval.
Arkin et al.~\cite{arkin1991efficiently} represent polygons by their \enquote{turning function} - a function which follows the angle between the polygon's tangent and the $ x $-axis while traversing the perimeter of the polygon.
They define the distance between polygons to be variations of the $ L_p $ (for $p=1,2$) distance between their turning functions.
This metric is invariant under translation, rotation and scaling (and the selection of the initial point on the perimeter) and therefore models well the intuitive notion of shape resemblance.

We develop and analyze LSH near neighbor data structures for several variations of the $ L_p $ distance for functions (for $p=1,2$). By applying our schemes to the turning functions of a collection of polygons we obtain efficient near neighbor LSH-based structures for polygons.
To tune our structures to turning functions of polygons, we prove some new properties of these turning functions that may be of independent interest.

As part of our analysis, we address the following problem which is of independent interest. Find the vertical translation of a function $ f $ that is closest in $ L_1 $ distance to a function $ g $. We prove tight bounds on the approximation guarantee obtained by the translation which is equal to the difference between the averages of $ g $ and $ f $.
	\end{abstract}
    \newpage
	\section{Introduction}
This paper focuses on similarity search between polygons, where we aim to efficiently retrieve polygons with a shape resembling the query polygon.

Large image databases are used in many multimedia applications in fields such as computer vision, pattern matching, content-based image retrieval, medical diagnosis and geographical information systems. Retrieving images by their content in an efficient and effective manner has therefore become an important task, which is of rising interest in recent years.

When designing content-based image retrieval systems for large databases, the following properties are typically desired:

\textbf{Efficiency}: Since the database is very large, iterating over all objects is not feasible, so an efficient indexing data structure is necessary.

\textbf{Human perception}: The retrieved objects should be perceptually similar to the query.

\textbf{Invariance to transformations}: The retrieval probability of an object should be invariant to translating, scaling, and rotating the object. Moreover, since shapes are typically defined by a time signal describing their boundary, we desire invariance also to the initial point of the boundary parametrization.

There are two general methods to define how much two images are similar (or distant): intensity-based (color and texture) and geometry-based (shape).
The latter method is arguably more intuitive~\cite{schomaker1999using} but more difficult since capturing the shape is a more complex task than representing color and texture features.
Shape matching has been approached in several other ways, including tree pruning~\cite{umeyama1993parameterized}, the generalized Hough transform~\cite{ballard1981generalizing}, geometric hashing~\cite{lamdan1988geometric} and Fourier descriptors~\cite{zahn1972fourier}. For an extensive survey on shape matching metrics see Veltkamp and Hagedoorn~\cite{veltkamp2001state}.

A noteworthy distance function between shapes is that of Arkin et al.~\cite{arkin1991efficiently}, which represents a curve using a cumulative angle function.
Applied to polygons, the \textit{turning function} (as used by Arkin et al.~\cite{arkin1991efficiently}) $ t_P $ of a polygon $ P $ returns the cumulative angle between the polygon's counterclockwise tangent at the point and the $ x $-axis, as a function of the fraction $ x $ of the perimeter (scaled to be of length 1) that we have traversed in a counterclockwise fashion. The turning function is a step function that changes at the vertices of the polygon, and either increases with left turns, or decreases with right turns (see Figure~\ref{fig-turning-function}). Clearly, this function is invariant under translation and scale of the polygon.

To find similar polygons based on their turning functions, we define the distance $L_p(P,Q)$ between polygons $P$ and $Q$ to be the $L_p$ distance between their turning functions $t_P(x)$ and $t_Q(x)$. That is
\[ L_p(P,Q)=\left(\int_{0}^{1}\abs{t_P(x)-t_Q(x)}^p \right)^{1/p}.\]

The turning function $ t_P(x) $ depends on the rotation of $ P $, and the (starting) point of $ P $ where we start accumulating the angle. If the polygon is rotated by an angle $ \alpha $, then the turning function $ t_P(x) $ becomes $ t_P(x)+\alpha $. Therefore, we define the (rotation invariant) distance $D^{\updownarrow}_p(P,Q)$ between polygons $P$ and $Q$ to be the $D^{\updownarrow}_p$ distance between their turning functions $t_P$ and $t_Q$, which is defined as follows
\[ D^{\updownarrow}_p(P,Q)\stackrel{def}{=}D^{\updownarrow}_p(t_P,t_Q)\stackrel{def}{=}\min_{\alpha\in \reals}L_p(t_P+\alpha, t_Q)= \min_{\alpha\in \reals}\mysqrt{p}{\int_{0}^{1} \abs{t_P(x)+\alpha-t_Q(x)}^p dx}. \]

If the starting point of $ P $ is clockwise shifted along the boundary by $ t $, the turning function $ t_P(x) $ becomes $ t_P(x+t) $. Thus, we define the distance $ D_p(P,Q) $ between polygons $ P $ and $ Q $ to be the $ D_p $ distance between their turning functions $t_P$ and $t_Q$ which is defined as follows
\begin{align*}
 D_p(P,Q)\stackrel{def}{=}D_p(t_P,t_Q)\stackrel{def}{=} \min_{\alpha \in \reals,t\in[0,1]}\left(\int_{0}^{1}\abs{t_P(x+t)+\alpha-t_Q(x)}^p \right)^{1/p}.
\end{align*}

The distance $ D_p(f,g) $ between two functions $ f $ and $ g $ extends $ f $ to the domain $ [0,2] $ by defining $ t_P(x+1)=t_P(x)+2\pi $. The distance metric $ D_p $ is invariant under translation, rotation, scaling and the selection of the starting point.
A comprehensive presentation of these distances, as well as a proof that they indeed satisfy the
metric axioms appears in \cite{arkin1991efficiently}.

We develop efficient nearest neighbor data structures for functions under these distances and then specialize them to functions which are turning functions of polygons.

Since a major application of polygon similarity is content-based image retrieval from large databases (see Arkin et al.~\cite{arkin1991efficiently}), the efficiency of the retrieval is a critical metric.
Traditionally, efficient retrieval schemes used tree-based indexing mechanisms, which are known to work well for prevalent distances (such as the Euclidean distance) and in low dimensions. Unfortunately such methods do not scale well to higher dimensions and do not support more general and computationally intensive metrics. To cope with this phenomenon (known as the \enquote{curse of dimensionality}), Indyk and Motwani~\cite{indyk1998approximate,har2012approximate} introduced Locality Sensitive Hashing (LSH), a framework based on hash functions for which the probability of hash collision is higher for near points than for far points.

Using such hash functions, one can determine near neighbors by hashing the query point and retrieving the data points stored in its bucket. Typically, we concatenate hash functions to reduce false positives, and use several hash functions to reduce false negatives. This gives rise to a data structure which satisfies the following property: for any query point $ q $, if there exists a neighbor of distance at most $ r $ to $q$ in the database, it retrieves (with constant probability) a neighbor of distance at most $cr $ to $q$ for some constant $c>1$. This data structure is parameterized by the parameter $ \rho =\frac{\log (p_1)}{\log (p_2)}<1 $, where $ p_1 $ is the minimal collision probability for any two points of distance at most $ r $, and $ p_2 $ is the maximal collision probability for any two points of distance at least $ cr $.
The data structure can be built in time and space $O(n^{1+\rho})$, and its query time is $ O(n^\rho \log_{1/p_2}(n)) $ where $ n $ is the size of the data set.\footnote{To ease on the reader, in this paper we suppress the term $ 1/p_1 $ in the structure efficiency, and the time it takes to compute a hash and distances between two polygons/functions. For example for polygons with at most $ m $ vertices (which we call $ m $-gons), all our hash computations take $ O(m) $ time, and using Arkin et al.~\cite{arkin1991efficiently} we may compute distances in $ O(m^2\log(m)) $ time.}

The trivial retrieval algorithm based on the turning function distance of Arkin et al.~\cite{arkin1991efficiently}, is to directly compute the distance $ D_2(P,Q) $ (or $ D_1(P,Q) $) between the query $ Q $ and all the polygons $ P $ in the database. This solution is invariant to transformations but not efficient (i.e., linear in the size of the database).

In this paper, we rely on the turning function distance of Arkin et al.~\cite{arkin1991efficiently} for $ p=1,2 $, and create the first retrieval algorithm with respect to the turning function distance which is sub-linear in the size of the dataset. To do so, we design and analyze LSH retrieval structures for function distance, and feed the turning functions of the polygons to them.
Our results give rise to a shape-based content retrieval (a near neighbor polygon) scheme which is efficient, invariant to transformations, and returns perceptually similar results.

\subsection*{Our contribution}
\begin{figure}[ht]
	\centering
	\includegraphics[width=\linewidth]{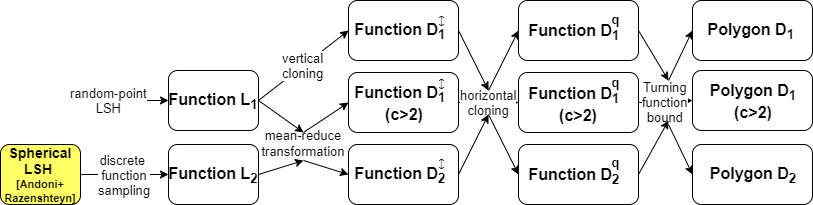}
	\caption{Our structures: each box is an $ (r,cr) $-\LSH near neighbor data structure, and the arrow $ A\to B $ with label $ t $ signifies that we use the method $ t $ over the structure $ A $ to get a structure for $ B $.}
	\label{fig:results}
\end{figure}

We develop simple but powerful $ (r,cr) $-LSH near neighbor data structures for efficient similar polygon retrieval, and give a theoretical analysis of their performance. We give the first structure (to the best of our knowledge) for approximate similar polygon retrieval which is provably invariant to shape rotation, translation and scale, and with a query time which is sub-linear in the number of data polygons. In contrast to many other structures for similar shape retrieval which often use heuristics, all our results are backed with theoretical proofs, using properties of the turning function distance and the theory of LSH.

To give our $ (r,cr) $-LSH near neighbor data structures for polygons, we build such structures for step functions with distances which are derived from the $ L_p $ distance for $ p=1,2 $, and apply them to turning functions of polygons.\footnote{Our structures for step functions can be extended to support also functions which are concatenations of at most $ k\in \naturals $ functions which are $ M $-Lipschitz for some $ M>0 $. Also, we can give similar structures for variations of the function $ D_1 $ and $ D_2 $ distances where we extend the functions from the domain $ [0,1] $ to the domain $ [0,2] $, not by $f(x)=f(x-1)+ 2\pi $, but by $f(x)=f(x-1)+q $ for any constant $ q\in\reals $.}
Here $ r>0 $ and $ c>1 $ are the LSH parameters as defined above, and $ n $ is the number of objects in the data structure.
The $ (r,cr) $-LSH data structures which we present exist for any $ r>0 $ and $ c>1 $ (except when $ c $ is explicitly constrained).
For an interval $ I $, we say that a function $ f:I\to \reals $ is a $ k $-step function, if $ I $ can be divided into $ k $ sub-intervals, such that over each sub-interval $ f $ is constant. All the following results for functions are for $ k $-step functions with ranges bounded in $ [a,b] $ for some $ a<b $ where for simplicity of presentation, we fix $ a=0 $ and $ b=1 $.\footnote{For general values of these parameters, the dependency of the data structure's run-time and memory is roughly linear or squared in $ b-a $.}$ ^, $\footnote{Since $ a=0 $ and $ b=1 $, the distance between any two functions is at most 1, so we focus on $ r<1 $.} The results we present below are slightly simplified versions than those that appear in the body of the paper. For an overview of our structures see Figure~\ref{fig:results}.

\paragraph*{Near neighbors data structures for functions}

\subparagraph*{1.}
For the $ L_1 $ distance over functions, we design a simple but powerful $ \LSH $ hash family. This hash selects a uniform point $ p $ from the rectangle $ [0,1]\times [0,1] $, and maps each function to 1, 0 or $ -1 $ based on its vertical relation (above, on or below) with $ p $. This yields an $ (r,cr) $-LSH structure for $ L_1 $ which requires sub-quadratic preprocessing time and space of $ O(n^{1+\rho}) $, and sub-linear query time of $ O(n^\rho \log{n}) $, where $ \rho=\log(1-r)\big/ \log(1-cr)\leq\frac{1}{c}$.
For the $ L_2 $ distance over functions, we observe that sampling each function at evenly spaced points reduces the $ L_2 $ distance to Euclidean distance. We use the data structure of Andoni and Razenshteyn~\cite{andoni2015optimal} for the Euclidean distance to give an $ (r,cr) $-LSH for the $ L_2 $ distance, which requires sub-quadratic preprocessing time of $ O(n^{1+\rho}+n_{r,c}\cdot n) $, sub-quadratic space of $ O(n_{r,c}\cdot n^{1+\rho}) $ and sub-linear query time of $ O(n_{r,c}\cdot n^{\rho}) $, where $ \rho=\frac{1}{2c-1}$ and $ n_{r,c}=\frac{2k}{(\sqrt{c}-1)r^2} $ is the dimension of the sampled vectors. We also give an alternative asymmetric $\LSH$ hash family for the $L_2$ distance inspired by our hash family for the $L_1$ distance, and create an $\LSH$ structure based on it.

\subparagraph*{2.}
For the $ D_2^\updownarrow $ distance, we leverage a result of Arkin et al.~\cite{arkin1991efficiently}, to show that the mean-reduce transformation, defined to be $ \hat{\phi}(x)=\phi(x)-\int_{0}^{1}\phi(s)ds $, reduces $ D_2^\updownarrow $ distances to $ L_2 $ distances with no approximation loss. That is, for every $ f $ and $ g $, $D_2^{\updownarrow}(f,g)=L_2(\hat{f}, \hat{g}) $, so we get an $ (r,cr) $-\LSH structure for the $ D_2^\updownarrow $ distance which uses our previous $L_2 $ structure, and with identical performance. For the $ D_1^\updownarrow $ distance, we approximately reduce $ D_1^\updownarrow $ distances to $ L_1 $ distances using the same mean-reduction. We give a simple proof that this reduction gives a 2-approximation, and improve it to a tight approximation bound showing that for any two step functions $ f,g:[0,1]\to [0,1] $, $L_1(\hat{f},\hat{g})\leq \left(2-D_1^{\updownarrow}(f,g)\right)\cdot D_1^{\updownarrow}(f,g)$. This proof (see full version), which is of independent interest, characterizes the approximation ratio by considering the function $ f-g $, dividing its domain into 3 parts and averaging over each part, thereby considering a single function with 3 step heights. This approximation scheme yields an $ (r,cr) $-LSH structure for any $ c>2-r $, which is substantially smaller than $ 2 $ (approaching $ 1 $) for large values of $ r $.

We also give an alternative structure \sslsh that supports any $ c>1 $, but has a slightly diminished performance.
This structure leans on the observation of Arkin et al.~\cite{arkin1991efficiently}, that the optimal vertical shift aligns a step of $ f $ with a step of $ g $. It therefore replaces each data step function by a set of vertical shifts of it, each aligning a different step value to $ y=0 $, and constructs an $ L_1 $ data structure containing all these shifted functions. It then replaces a query with its set of shifts as above, and performs a query in the internal $ L_1 $ structure with each of these shifts.

\subparagraph*{3.}
For the $ D_1$ and  $ D_2$ distances, we leverage another result of Arkin et al.~\cite{arkin1991efficiently}, that the optimal horizontal shift horizontally aligns a discontinuity point of $ f $ with a discontinuity point of $ g $. Similarly to \sslsh, we give a structure for $ D_1$ (or $ D_2$) by keeping an internal structure for $ D_1^{\updownarrow} $ (or $ D_2^\updownarrow $) which holds a set of horizontal shifts of each data functions, each aligns a different discontinuity point in to $ x=0 $. It then replaces a query with its set of shifts as above, and performs a query in the internal structure with each of these shifts.

\paragraph*{Near neighbors data structures for polygons}
We design LSH structures for the polygonal $ D_1 $ and $ D_2 $ distances, by applying the $ D_1 $ and $ D_2$ structures to the turning functions of the polygons. We assume that all the data and query polygons have at most $ m $ vertices (are $ m $-gons), where $ m $ is a constant known at preprocessing time. It is clear that the turning functions are $(m+1) $-step functions, but the range of the turning functions is not immediate (note that performance inversely relates to the range size).

First, we show that turning functions of $ m $-gons are bounded in the interval $I=\left[-(\floor{m/2}-1)\pi, (\floor{m/2}+3)\pi\right]$ of size $ \lambda_m:=(2\cdot \floor{m/2}+2)\pi $. We show that this bound is tight in the sense that there are $ m $-gons whose turning functions get arbitrarily close to these upper and lower bounds.

Second, we define the $ span $ of a function $ \xi:[0,1]\to \reals $ to be $ span(\xi)=\max_{x\in [0,1]}(\xi(x))-\min_{x\in [0,1]}(\xi(x)) $, and show that for $ m $-gons, the span is at most $\lambda_m/2=(\floor{m/2}+1)\pi  $, and that this bound is tight - there are $ m $-gons whose turning functions have arbitrarily close spans to $ \lambda_m/2 $. Since the $ D_1$ and $ D_2$ distances are invariant to vertical shifts, we perform an a priori vertical shift to each turning function such that its minimal value becomes 0, effectively morphing the range to $ [0,\lambda_m/2] $, which is half the original range size. This yields the following structures:

For the $ D_1 $ distance, for any $ c>2$ we give an $ (r,cr) $-LSH structure storing $ n $ polygons with at most $ m $ vertices which requires $ O((nm)^{1+\rho}) $ preprocessing time and space which are sub-quadratic in $ n $, and $ O(m^{1+\rho}n^{\rho} \log (nm)) $ query time which is sub-linear in $ n $, where $ \rho$ is roughly $ 2/c$. Also for $ D_1 $, for any $ c>1$ we get an $ (r,cr) $-LSH structure which requires sub-quadratic preprocessing time and space of $ O((nm^2)^{1+\rho}) $, and sub-linear query time of $ O(m^{2+2\rho}n^{\rho} \log (nm)) $, where $ \rho$ is roughly $ 1/c$.

For the $ D_2 $ distance, we give an $ (r,cr) $-LSH structure which requires sub-quadratic preprocessing time of $ \tilde{O}(n^{1+\rho}) $, sub-quadratic space of $ \tilde{O}(n^{1+\rho}) $, and sub-linear query time of $ \tilde{O}(n^{\rho}) $, where $ \rho = \frac{1}{2\sqrt{c}-1}$.\footnote{The $ \tilde{O} $ notation hides multiplicative constants which are small powers (e.g., $5$) of $ m $, $\frac{1}{r}$ and $ \frac{1}{\sqrt[4]{c}-1} $.}

\subsection*{Other similar works}
Babenko et al.~\cite{babenko2014neural} suggest a practical method for similar image retrieval, by embedding images to a Euclidean space using Convolutional Neural Networks (CNNs), and retrieving similar images to a given query based on their embedding's euclidean distance to the query embedding. This approach has been the most effective practical approach for similar image retrieval in recent years.

Gudmundsson and Pagh~\cite{gudmundsson2017range} consider a metric in which there is a constant grid of points, and shapes are represented by the subset of grid points which are contained in them. The distance between polygons is then defined to be the \textit{Jaccard distance} between the corresponding subsets of grid points. Their solution lacks invariance to scale, translation and rotation, however our work is invariant to those, and enables retrieving polygons which have a similar shape, rather than only spatially similar ones.

Other metrics over shapes have been considered. Cakmakov et al.~\cite{cakmakov2004estimation} defined a metric based on snake-like moving of the curves. Bartolini et al.~\cite{bartolini2002using} proposed a new distance function between shapes, which is based on the Discrete Fourier Transform and the Dynamic Time Warping distance. Chavez et al.~\cite{chavez2016affine} give an efficient polygon retrieval technique based on Fourier descriptors. Their distance works for exact matches, but is a weak proxy for visual similarity, since it relates to the distances between corresponding vertices of the polygons.

There has been a particular effort to develop efficient structures for the discrete Fréchet distance and the dynamic time warping distance for polygonal curves in $ \reals^d $. Such works include Driemel et al.~\cite{driemel2017locality} who gave LSH structures for these metrics via snapping the curve points to a grid, Ceccarello et al.~\cite{ceccarello2019fresh} who gave a practical and efficient algorithm for the r-range search for the discrete Fréchet distance, Filtser et al.~\cite{filtser2019approximate} who built a deterministic approximate near neighbor data structure for these metrics using a subsample of the data, and Astefanoaei et al.~\cite{astefanoaei2018multi} who created a suite of efficient sketches for trajectory data.
Grauman and Darrell~\cite{grauman2004fast} performed efficient contour-based shape retrieval (which is sensitive (not invariant) to translations, rotations and scaling) using an embedding of Earth Mover’s Distance into $ L_1 $ space and LSH.

\section{Preliminaries}
We first formally define LSH, then discuss the turning function representation of Arkin et al.~\cite{arkin1991efficiently}, and then define the distance functions between polygons and functions which rise from this representation.
\subsection{Locality sensitive hashing}
We use the following standard definition of a \textit{Locality Sensitive Hash Family (LSH)} with respect to a given distance function $ d:Z\times Z\to \reals_{\geq 0} $.
\begin{definition}[Locality Sensitive Hashing (LSH)]\label{def:LSH}
	Let $ r>0$, $c > 1$ and $ p_1>p_2 $. A family $ H $ of functions $ h:Z\to \Gamma $ is an $ (r, cr, p_1, p_2) $-\LSH for a distance function $d:Z\times Z \to \reals_{\geq 0}$ if for any $ x,y\in Z$,
	\begin{enumerate}
		\item If $ d(x, y) \leq r $ then $ \Pr_{h\in H}[h(x)=h(y)] \geq p_1 $, and
		\item If $ d(x, y) \geq cr $ then $ \Pr_{h\in H}[h(x)=h(y)] \leq p_2 $.
	\end{enumerate}
\end{definition}
Note that in the definition above, and in all the following definitions, the hash family $H$ is always sampled uniformly.

We say that a hash family is an \textit{$(r, cr)$-\LSH} for a distance function $ d $ if there exist $ p_1 > p_2  $ such that it is an $ (r, cr, p_1, p_2) $-\LSH. A hash family is a \textit{universal LSH} for a distance function $ d $ if for all $ r > 0 $ and $ c>1 $ it is an $ (r, cr) $-\LSH.

From an $ (r, cr, p_1, p_2) $-\LSH family, we can derive, via the general theory developed in \cite{indyk1998approximate,har2012approximate}, an \textit{$(r,cr)$-\LSH data structure}, for finding approximate near neighbors with respect to $r$. That is a data structure that finds (with constant probability) a neighbor of distance at most $cr$ to a query $q$ if there is a neighbor of distance at most $ r $ to $q$. This data structure uses $ O(n^{1+\rho}) $ space (in addition to the data points), and $ O(n^\rho \log_{1/p_2}(n)) $ hash computations per query, where $ \rho=\frac{\log (1/p_1)}{\log (1/p_2)} =\frac{\log (p_1)}{\log (p_2)}$.

\subsection{Representation of polygons}
\begin{figure}[ht]
	\centering
	\includegraphics[width=0.8\linewidth]{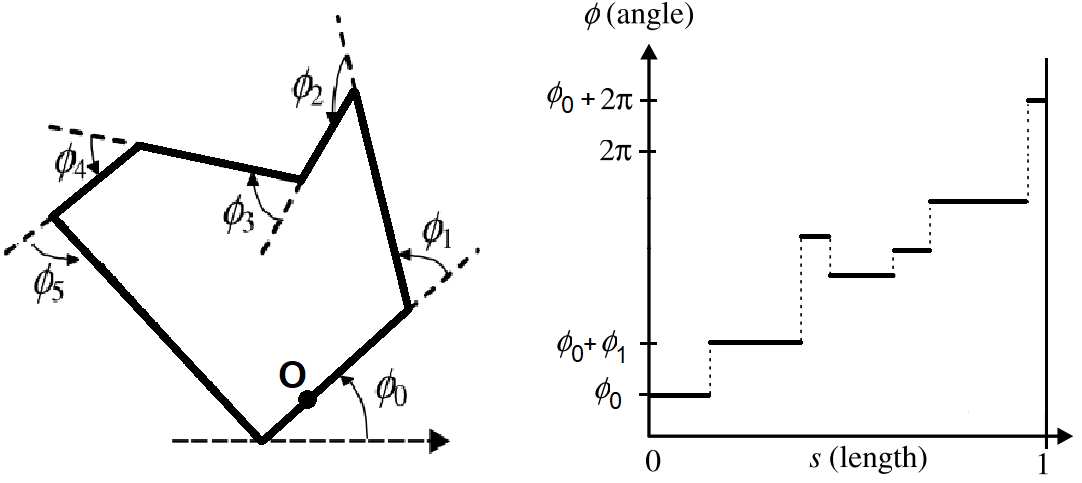}
	\caption{Left: a polygon $ P $ with $ 6 $ vertices. Right: the turning function $ t_P $ of $ P $, with $ 7 $ steps. }
	\label{fig-turning-function}
\end{figure}

Let $ P $ be a simple polygon scaled such that its perimeter is one. Following the work of Arkin et al.~\cite{arkin1991efficiently}, we represent $ P $ via a \textit{turning function} $ t_P(s):[0,1]\to \reals$, that specifies the angle of the counterclockwise tangent to $ P $ with the x-axis, for each point $ q $ on the boundary of $ P $. A point $ q $ on the boundary of $ P $ is identified by its counterclockwise distance (along the boundary which is of length 1 by our scaling) from some fixed reference point $ O $. It follows that $ t_P(0) $ is the angle $ \alpha $ that the tangent at $ O $ creates with the x-axis, and $ t_P(s) $ follows the cumulative turning, and increases with left turns and decreases with right turns.
Although $ t_P $ may become large or small, since $ P $ is a simple closed polygon we must have that $ t_P(1)=t_P(0)+2\pi $ if $ O $ is not a vertex of $ P $, and $ t_P(1)-t_P(0)\in [\pi,3\pi] $ otherwise.
Figure~\ref{fig-turning-function} illustrates the polygon turning function.

Note that since the angle of an edge with the x-axis is constant and angles change at the vertices of $ P $, then the function is constant over the edges of $ P $ and has discontinuity points over the vertices. Thus, the turning function is in fact a step function.

In this paper, we often use the term $ m $-gon — a polygon with \textbf{at most} $ m $ vertices.

\subsection{Distance functions}
Consider two polygons $ P$ and $Q $, and their associated turning functions $ t_P(s) $ and $ t_Q(s) $ accordingly. Define the \textit{aligned $ L_p $ distance} (often abbreviated to \textit{$ L_p $ distance}) between $ P $ and $ Q $ denoted by $ L_p(P,Q) $, to be the $ L_p $ distance between $ t_P(s) $ and $ t_Q(s) $ in $ [0,1] $:
$ L_p(P,Q)=\mysqrt{p}{\int_{0}^{1} \abs{t_P(x)-t_Q(x)}^p dx}$.

Note that even though the $ L_p $ distance between polygons is invariant under scale and translation of the polygon, it depends on the rotation of the polygon and the choice of the reference points on the boundaries of $ P $ and $ Q $.

Since rotation of the polygon results in a vertical shift of the function $t_P$, we define the \textit{vertical shift-invariant} $ L_p $ distance between two functions $ f $ and $ g $ to be\\
$ D^{\updownarrow}_p(f,g)=\min_{\alpha\in \reals}L_p(f+\alpha, g)= \min_{\alpha\in \reals}\mysqrt{p}{\int_{0}^{1} \abs{f(x)+\alpha-g(x)}^p dx} $.
Accordingly, we define the \textit{rotation-invariant $ L_p $} distance between two polygons $ P $ and $ Q $ to be the vertical shift-invariant $ L_p $ distance between the turning functions $ t_P $ and $ t_Q $ of $ P $ and $ Q $ respectively:
$ D^{\updownarrow}_p(P,Q)=D^{\updownarrow}_p(t_P,t_Q)=\min_{\alpha\in \reals}\mysqrt{p}{\int_{0}^{1} \abs{t_P(x)+\alpha-t_Q(x)}^p dx}.$

To tweak the distance $ D^{\updownarrow}_p $ such that it will be invariant to changes of the reference points, we need the following definition. We define the \textit{$ 2\pi $-extension} $ f^{2\pi}:[0,2]\to \reals $ of a function $f:[0,1]\to \reals $ to the domain $ [0,2] $, to be
$f^{2\pi}=\begin{cases}
f(x), &\quad \text{for }  x\in[0,1]\\
f(x-1)+2\pi, &\quad \text{for }  x\in(1,2]\\
\end{cases}.$

A turning function $t_P $ is naturally $ 2\pi $-extended to the domain $ [0,2] $ by circling around $ P $ one more time.
We define the $ u $-slide of a function $g:[0,2]\to \reals $, $ slide^{\leftrightarrow}_{u}(g):[0,1]\to \reals $, for a value $ u\in [0,1] $ to be $(slide^{\leftrightarrow}_{u}(g))(x)=g(x+u)$.
These definitions are illustrated in Figure~\ref{fig:ushifts}.
Note that shifting the reference point by a counterclockwise distance of $ u $ around the perimeter of a polygon $ P $ changes the turning function from $ t_P $ to $ slide^{\leftrightarrow}_{u}(t_P^{2\pi}) $.

\begin{figure}[ht]
	\centering
	\includegraphics[width=\linewidth]{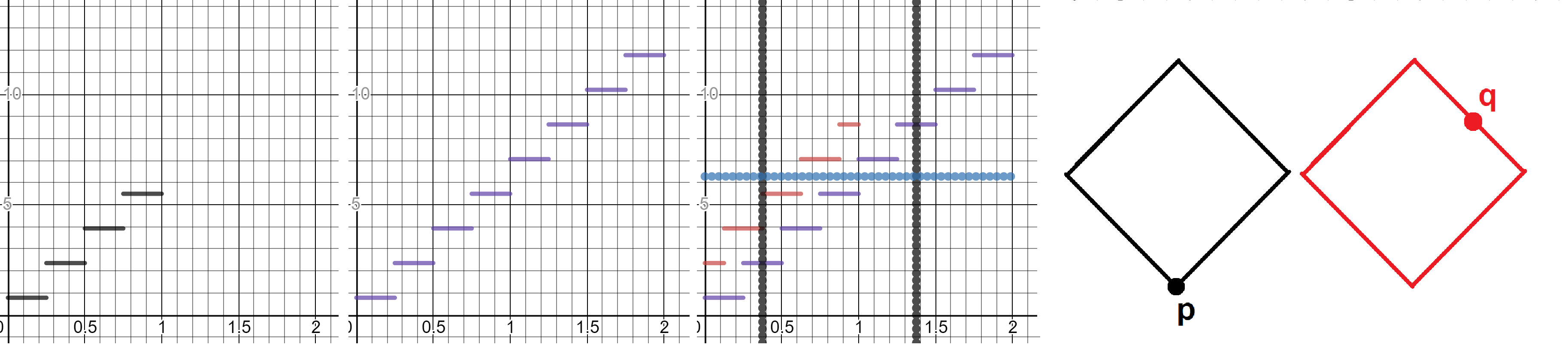}
	\caption{Left: The turning function $ t_P $ of the square with reference point $ p $. Center: the $ 2\pi $-extension $ t_P^{2\pi} $ of $ t_P $. Right: The turning function of the square with the reference point $ q $ in red (this is in fact the function $ t_P^{2\pi} $ cropped to between the black vertical lines, i.e., to $ [0.375,1.375] $).}
	\label{fig:ushifts}
\end{figure}

We therefore define the (vertical and horizontal) \textit{shift-invariant $ L_p $ distance} between two functions $ f ,g:[0,1]\to \reals $ to be:\\
$ D_{p}(f,g)=\min_{u\in [0,1]}D^{\updownarrow}_p(slide^{\leftrightarrow}_{u}(f^{2\pi}),g) =\min_{\alpha\in \reals,~u\in [0,1]}\mysqrt{p}{ \int_{0}^{1}\abs{f^{2\pi}(x+u)+\alpha-g(x)}^p dx},
$
and define the (rotation and reference point invariant) $ L_p $ distance between two polygons $ P $ and $ Q $ to be
$ D_p(P,Q)=D_{p}(t_P,t_Q)$. Arkin et al.~\cite{arkin1991efficiently} proved that $D_{p}(f,g) $ is a metric for any $ p>0 $.

\section{\texorpdfstring{$ L_1 $-}{}based distances}\label{sec:l1Distances}
In this section, we give LSH structures for the $ L_1 $ distance, the $ D_1^\updownarrow $ distance and then the $ D_1$ distance.
Note that the $ D_1$ distance reduces to the $ D_1^{\updownarrow} $ distance, which by using the \textit{mean-reduction} transformation presented in Section~\ref{subsec:d1updown}, reduces to the $ L_1 $ distance.

\subsection{Structure for \texorpdfstring{$ L_1 $}{}}\label{sec:hashdelta1}
In this section we present \rplsh, a simple hash family for functions $ f:[0,1]\to [a,b] $ with respect to the $ L_1 $ distance. \RPLSH is the hash family
$ H_1(a,b)=\left\{h_{(x,y)} \mid (x,y)\in [0,1]\times [a,b]\right\}$,
where the points $ (x,y) $ are uniformly selected from the rectangle $ [0,1]\times [a,b] $.
Each $ h_{(x,y)} $ receives a function $ f: [0,1]\to [a,b]$, and returns $ 1 $ if $ f $ is vertically above the point $(x,y) $, returns $ -1 $ if $ f $ is vertically below $ (x,y) $, and $0$ otherwise.

\begin{figure}[ht]
	\centering
	\includegraphics[width=0.4\linewidth]{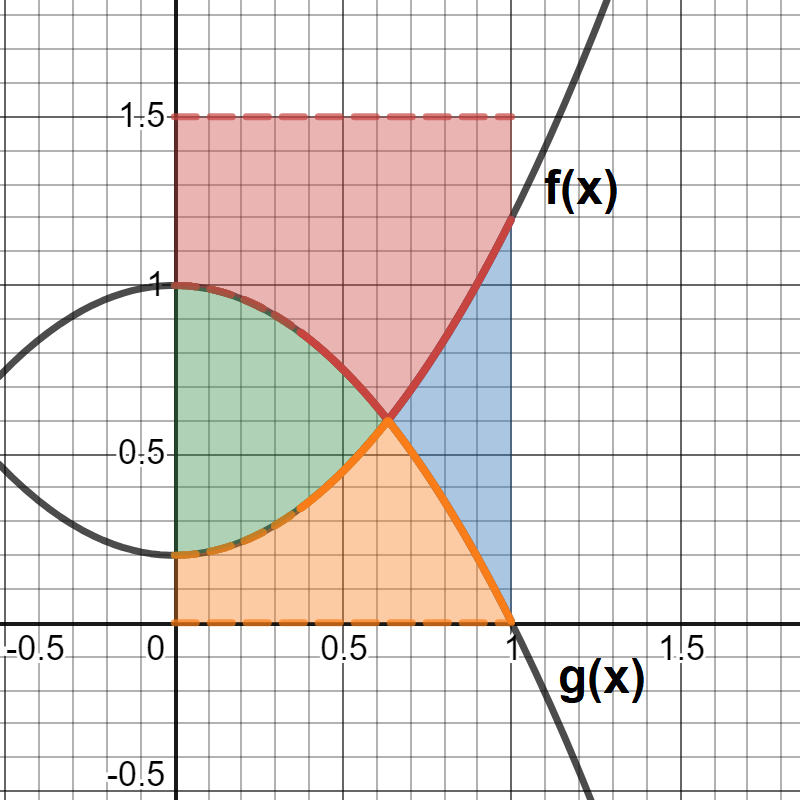}
	\caption{Illustration of the hash of two functions $ f $ and $ g $ w.r.t.\ $ h_{(x,y)} $ for $ a=0 $ and $ b=1.5 $. For $ (x,y) $ in the green area $ h_{(x,y)}(f)=-1\ne 1 =h_{(x,y)}(g)$, in the blue area $ h_{(x,y)}(f)=1\ne-1=h_{(x,y)}(g)$, in the red area $ h_{(x,y)}(f)=h_{(x,y)}(g) =-1$, and in the orange area $ h_{(x,y)}(f)=h_{(x,y)}(g) =1$. }
	\label{fig:uplsh}
\end{figure}

The intuition behind \rplsh is that any two functions $ f,g:[0,1]\to [a,b]$ collide precisely over hash functions $ h_{(x,y)} $ for which the point $ (x,y) $ is outside the area bounded between the graphs of $ f $ and $ g $. This fact is illustrated in the following Figure~\ref{fig:uplsh}.
Thus, this hash incurs a collision probability of $ 1-\frac{L_1(f,g)}{b-a}= 1-\frac{L_1(f,g)}{b-a}$, which is a decreasing function with respect to $L_1(f,g)$. This intuition leads to the following results.

\begin{theorem}\label{thm:L1hashCollisionProb}
	For any two functions $ f,g:[0,1]\to [a,b] $, we have that $P_{h\sim H_1(a,b)}(h(f)=h(g)) =1-\frac{L_1(f,g)}{b-a}.$
\end{theorem}
\begin{proof}

	Fix $ x\in [0,1] $, and denote by $ U(S) $ the uniform distribution over a set $ S $. We have that
	\begin{align*}
	P_{y\sim U([a,b])}(h_{(x,y)}(f)&=h_{(x,y)}(g))=1- P_{y\sim U([a,b])}(h_{(x,y)}(f)\ne h_{(x,y)}(g))\\
	&=1-\frac{\abs{f(x)-g(x)}}{b-a},
	\end{align*}
	where the last equality follows since $ h_{(x,y)}(f)\ne h_{(x,y)}(g) $ precisely for the $ y $ values between $ f(x) $ and $ g(x) $.
	Therefore, by the law of total probability,
	\begin{align*}
	P_{h\sim H_1(a,b)}(h(f)=h(g))&=P_{(x,y)\sim U([0,1]\times [a,b])}(h_{(x,y)}(f)=h_{(x,y)}(g))\\
	&=\int_{0}^{1} P_{y\sim U([a,b])}(h_{(x,y)}(f)=h_{(x,y)}(g)) dx\\
	&=\int_{0}^{1}\left(1- \frac{\abs{f(x)-g(x)}}{b-a}\right) dx
	=1-\frac{L_1(f,g)}{b-a}.\qedhere
	\end{align*}
\end{proof}

\begin{corollary}\label{cor:delta1lshstruct}
		For any $ r>0 $ and $ c>1 $, one can construct an $ (r,cr) -$LSH structure for the $ L_1 $ distance for $ n $ functions with ranges bounded in $ [a,b] $. This structure requires $ O(n^{1+\rho}) $ space and preprocessing time, and has $ O(n^\rho \log(n)) $ query time, where $ \rho=\frac{\log \left(1-\frac{r}{b-a}\right)}{\log \left(1-\frac{cr}{b-a}\right)}\approx \frac{1}{c} $ for $ r\ll b-a $.
\end{corollary}
\begin{proof}
	Fix $ r> 0 $ and $ c>1 $. By the general result of Indyk and Motwani~\cite{indyk1998approximate}, it suffices to show that $ H_1(a,b) $ is an $ (r,cr, 1-\frac{r}{b-a},1- \frac{cr}{b-a}) $-\LSH for the  $ L_1 $ distance.

	Indeed, by Theorem~\ref{thm:L1hashCollisionProb}, $P_{h\sim H_1(a,b)}(h(f)=h(g)) =1-\frac{L_1(f,g)}{b-a} $, so we get that
	\begin{itemize}
		\item If $  L_1 (f,g)\leq r $, then $ P_{h\sim H_1(a,b)}(h(f)=h(g))=1-\frac{L_1(f,g)}{b-a}\geq 1-\frac{r}{b-a}.$
		\item If $  L_1 (f,g)\geq cr $, then $ P_{h\sim H_1(a,b)}(h(f)=h(g))=1-\frac{L_1(f,g)}{b-a}\leq 1-\frac{cr}{b-a}.$
	\end{itemize}
\end{proof}

\subsection{Structure for \texorpdfstring{$ D_1^\updownarrow $}{}}\label{subsec:d1updown}

In this section we present \mrlsh, an \LSH family for the vertical translation-invariant $ L_1 $ distance, $ D_1^{\updownarrow} $.
Observe that finding an \LSH family for $ D_1^{\updownarrow} $ is inherently more difficult than for $ L_1 $, since even evaluating $ D_1^{\updownarrow}(f,g) $ for a query function $ g $ and an input function $ f $ requires minimizing $L_1(f+\alpha,g)$ over the variable $ \alpha $, and the optimal value of $ \alpha $ depends on both $ f $ and $ g $.

Our structure requires the following definitions. We define $ \bar{\phi}=\int_{0}^{1} \phi(x) dx $ to be the mean of a function $\phi$ over the domain $ [0,1] $, and define the \textit{mean-reduction} of $ \phi $, denoted by $ \hat{\phi}:[0,1] \to [a-b,b-a] $, to be the vertical shift of $ \phi $ with zero integral over $ [0,1] $, i.e., $ \hat{\phi}(x)= \phi(x)-\bar{\phi}(x)$. These definitions are illustrated in Figure~\ref{fig:mean-reduce-illustration}.
Our solution relies on the crucial observation that for the pair of functions $ f,g:[0,1]\to [a,b] $, the value of $ \alpha $ which minimizes $L_1(f+\alpha,g)  $ is \enquote{well approximated} by $ \bar{g}-\bar{f} $. That is the distance $ L_1(f+(\bar{g}-\bar{f}),g)=L_1(f-\bar{f},g-\bar{g})=L_1(\hat{f},\hat{g}) $ approximates $ D_1^{\updownarrow}(f,g)$. This suggests that if we replace any data or query function $ f $ with $ \hat{f} $, then the $  D_1^{\updownarrow}  $ distances are approximately the  $ L_1 $ distances of the shifted versions $ \hat{f} $, for which we can use the hash $ H_1 $ from Section~\ref{sec:hashdelta1}.

\begin{figure}[ht]
	\centering
	\includegraphics[width=0.4\linewidth]{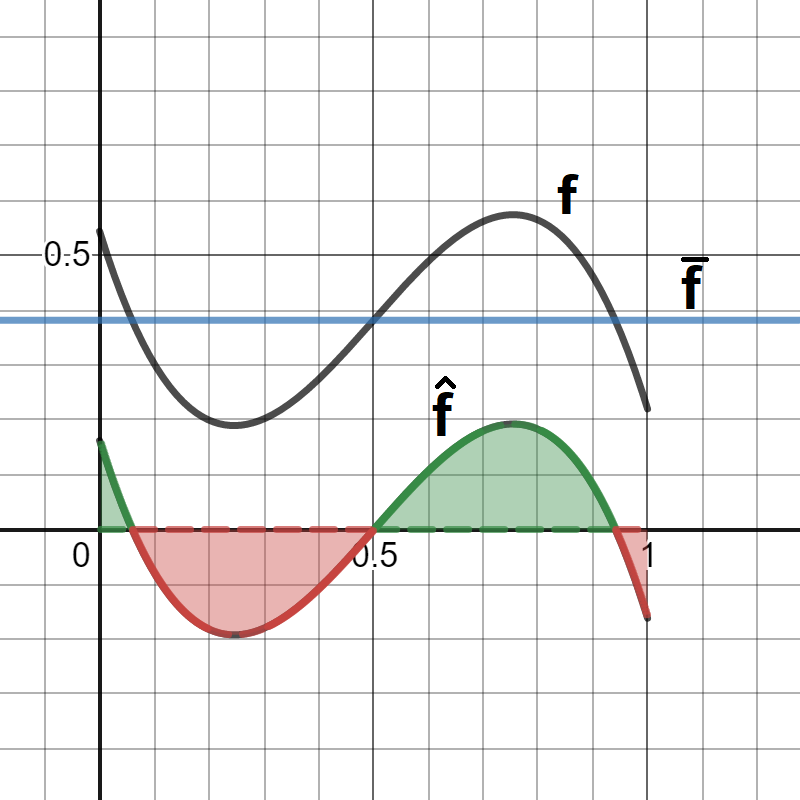}
	\caption{A function $ f$ (black), its mean $ \bar{f} $(blue), and its mean-reduction $ \hat{f} $ (below). Notice that the red and green areas are equal.}
	\label{fig:mean-reduce-illustration}
\end{figure}

Indeed, we use the hash family $ H_1 $ from Section~\ref{sec:hashdelta1}, and define \mrlsh for functions with images contained in $ [a,b] $ to be the family
$ H^{\updownarrow}_1(a,b)=\{f\to h\circ \hat{f} \mid h\in H_1(a-b,b-a)\}$.
Each hash of $ H^{\updownarrow}_1(a,b) $ is defined by a function $ h\in H_1(a-b,b-a) $, and given a function $ f $, it applies $ h $ on its mean-reduction $ \hat{f} $.

The following theorem gives a tight bound for the $ L_1 $ distance between mean-reduced functions in terms of their original vertical translation-invariant $ L_1 $ distance $ D_1^{\updownarrow} $. The proof of this tight bound as well as a simpler 2-approximation appear in Appendix~\ref{appndx:subsec:d1updown}.
Our elegant but more complicated proof of the tight bound characterizes and bounds the approximation ratio using properties of $ f-g $, and demonstrates its tightness by giving the pair of step functions $ f,g $ which meet the bound.

We conclude this result in the following theorem.

\begin{theorem}\label{thm:meanFunctionSubtractDistanceStrongerBoundWeak}
	Let $ f,g:[0,1]\to [a,b] $ be step functions and let $r\in (0,b-a]  $ be their vertical shift-invariant $ L_1 $ distance $ r=D_1^{\updownarrow}(f,g) $. Then
	$r\leq L_1(\hat{f},\hat{g})\leq \left(2-\frac{r}{b-a}\right)\cdot r.$ This bound is tight, i.e, there exist two functions $ f_0, g_0 $ as above for which $L_1(\hat{f_0},\hat{g_0})= \left(2-\frac{r}{b-a}\right)\cdot r.$
\end{theorem}

We use Theorem~\ref{thm:meanFunctionSubtractDistanceStrongerBoundWeak} to prove that \mrlsh is an \LSH family (Theorem~\ref{thm:H1updownisLsh}). We then use Theorem~\ref{thm:H1updownisLsh} and the general result of Indyk and Motwani~\cite{indyk1998approximate} to get Corollary~\ref{cor:lshstructD1updown}.

\begin{theorem}\label{thm:H1updownisLsh}
	For any $ r\in (0,b-a) $ and $ c>2-\frac{r}{b-a} $, $ H^{\updownarrow}_1(a,b)$ is an\\ $ \left(r,cr, 1-\left(2-\frac{r}{b-a}\right)\cdot\frac{r}{2(b-a)},1-c\cdot \frac{r}{2(b-a)}\right) $-LSH family for the $ D_1^{\updownarrow} $ distance.
\end{theorem}

\begin{corollary}\label{cor:lshstructD1updown}
	For any $ r> 0 $ and $ c>2-\frac{r}{b-a} $, one can construct an $ (r,cr) -$LSH structure for the $ D_1^{\updownarrow} $ distance for $ n $ functions with ranges bounded in $ [a,b] $. This structure requires $ O(n^{1+\rho}) $ extra space and preprocessing time, and $ O(n^\rho \log(n)) $ query time, where $ \tilde{r}=r/(2(b-a)) $ and
	$\rho =\log \left(1-(2-2\tilde{r})\cdot \tilde{r}\right)/\log \left(1-c\tilde{r}\right)$ for small $ \tilde{r} $.
\end{corollary}

\subsubsection*{\texorpdfstring{\SSLSH}{}}

We present \sslsh, a structure for the $ D_1^\updownarrow $ distance which works for any $ c>1 $ (unlike \mrlsh), but has a slightly worse performance, which depends on an upper bound $ k $ on the number of steps in of the data and query functions. This structure uses an internal structure for the $ L_1 $ distance, and leverages the observation of Arkin et al.~\cite{arkin1991efficiently} that the optimal vertical shift $ \alpha $ to align two step functions $ f $ and $ g $, is such that $ f+\alpha $ has a step which partially overlaps a step of $ g $, i.e., there is some segment $ S\subseteq [0,1] $ over which $ f+\alpha=g $.

Therefore, we overcome the uncertainty of the optimal $ \alpha $ by a priori cloning each function by the number of steps it has, and vertically shifting each clone differently to align each step to be at $ y=0 $.\footnote{This idea of cloning appears once again (but in a horizontal version), and in more detail, in Section~\ref{subsec:d1} for the $ D_1 $ distance.} For a query function $ g $, we clone it similarly to align each step to $ y=0 $, and use each clone as a separate query for the $ L_1 $ structure. This process effectively gives a chance to align each step of the query $ g $ with each step of each data step function $ f $.

\begin{corollary}\label{cor:lshstructD1updownverticalclones}
	For any $ a<b $, $ r>0 $ and $ c>1$, there exists an $ (r,cr) $-\LSH structure for the $ D_1^\updownarrow$ distance for $ n $ functions, each of which is a $ k $-step function with range bounded in $ [a,b] $. This structure requires $ O((nk)^{1+\rho})$ extra space and preprocessing time, and $  O(k^{1+\rho}n^\rho \log(nk))  $ query time, where
	$ \rho=\log \left(1-\frac{r}{2(b-a)}\right)/\log \left(1-\frac{cr}{2(b-a)}\right)\approx \frac{1}{c} $ for $ r\ll b-a $.
\end{corollary}

\subsection{Structure for \texorpdfstring{$ D_1 $}{}}\label{subsec:d1}

In this section, we present \sclsh, a data structure for the distance function $ D_{1}$ defined over step functions $ f:[0,1]\to [a,b] $.
To do so, we use an $ (r',c'r')$-\LSH data structure (for appropriate values of $ r' $ and $ c' $) for the distance function $ D_1^{\updownarrow} $ which will hold slided functions with ranges contained in $[a,b+2\pi] $.

Recall that the $ D_{1} $ distance between a data function $ f$ and a query function $ g $ is defined to be the minimal $ D^{\updownarrow}_1 $ distance between a function in the set $ \left\{slide^{\leftrightarrow}_{u}(f^{2\pi})\mid u\in [0,1]\right\} $ and the function $ g $, and we obviously do not know $ u $ a priori and cannot build a structure for each possible $u\in [0,1]$. Fortunately, in the proof of Theorem 6 from Arkin et al.~\cite{arkin1991efficiently}, they show that for any pair of step functions $ f $ and $ g $, the optimal slide $ u $ is such that a discontinuity of $ f $ is aligned with a discontinuity of $ g $. They show that this is true also for the $ D_2 $ distance.

Therefore, we can overcome the uncertainty of the optimal $ u $ by a priori cloning each function by the number of discontinuity points it has, and sliding each clone differently to align its discontinuity point to be at $ x=0 $. For a query function $ g $, we clone it similarly to align each discontinuity point to $ x=0 $, use each clone as a separate query. The above process effectively gives a chance to align each discontinuity point of the query function $ g $ with each discontinuity point of each data step function $ f $.

\SCLSH works as follows.

\paragraph*{Preprocessing phase}
We are given the parameters $ r>0 $, $ c>1 $, $ a<b $ and a set of step functions $ F $, where each function is defined over the domain $ [0,1] $ and has a range bounded in $ [a,b] $. Additionally, we are given an upper bound $ k $ on the number of steps a data or query step function may have. First, we replace each function $ f\in F $ with the set of (at most $ k+1 $) $ u $ slides of it's $ 2\pi $-extension for each discontinuity point $ u $, i.e., $ slide^{\leftrightarrow}_{u}(f^{2\pi}) $ for each discontinuity point $u\in [0,1]$. For each such clone we remember its original unslided function. Next, we store the at most $ (k+1)\cdot \abs{F} $ resulted functions in an $ (r',c'r')$-\LSH data structure for the $ D_1^{\updownarrow} $ distance for functions with ranges bounded in $[a,b+2\pi] $, tuned with the parameters $r'=r$ and $ c'=c $.

\paragraph*{Query phase}
Let $ g $ be a query function. We query the $ D_1^{\updownarrow}$ structure constructed in the preprocessing phase with each of the slided queries $ slide^{\leftrightarrow}_{u}(g^{2\pi}) $ for each discontinuity point $u\in [0,1]$. If one of the queries returns a data function $ f $, we return its original unslided function, and otherwise return nothing.

In Theorem~\ref{thm:DdistanceReduction}, we prove that \sclsh is an $ (r,cr) $-data structure for $ D_1$.
\begin{theorem}\label{thm:DdistanceReduction}
	\SCLSH is an $ (r,cr) $-\LSH structure for the $ D_1 $ distance.
\end{theorem}

\begin{corollary}\label{cor:lshstructD1}
	For any $ a<b $, $ r>0 $, $\omega=b+2\pi-a$ and $ c>2-\frac{r}{\omega} $, there exists an $ (r,cr) $-\LSH structure for the $ D_1$ distance for $ n $ functions, each of which is a $ k $-step function with range bounded in $ [a,b] $. This structure requires $ O((nk)^{1+\rho})$ extra space and preprocessing time, and $  O(k^{1+\rho}n^\rho \log(nk))  $ query time, where $ \tilde{r}=r/(2\omega) $ and
	$\rho =\log \left(1-(2-2\tilde{r})\cdot \tilde{r}\right)/\log \left(1-c\tilde{r}\right) \approx \frac{2}{c}$ for small $ \tilde{r} $.\footnote{Given a bound $ s $ on the span of the functions, we can a priori vertically shift all the functions such that their minimum is 0, effectively making the range size smaller (within $ [0,s] $) and improving the performance of the structure (see Appendix~\ref{sec:D1poly})}
\end{corollary}

\begin{corollary}\label{cor:lshstructD1withsslsh}
	For any $ a<b $, $ r>0 $ and $ c>1$, there exists an $ (r,cr) $-\LSH structure for the $ D_1$ distance for $ n $ functions, each of which is a $ k $-step function with range bounded in $ [a,b] $. This structure requires $ O((nk^2)^{1+\rho})$ extra space and preprocessing time, and $  O(k^{2+2\rho}n^\rho \log(nk))  $ query time, where
	$ \rho=\log \left(1-\frac{r}{2(b+2\pi-a)}\right)/\log \left(1-\frac{cr}{2(b+2\pi-a)}\right)\approx \frac{1}{c} $ for $ r\ll 2(b+2\pi-a) $.
\end{corollary}

	\section{\texorpdfstring{$ L_2 $-}{}based distances}\label{sec:l2distances}
This section, which appears in detail in Appendix~\ref{appndx:sec:l2distances}, gives LSH structures for the $ L_2 $ distance, the $ D_2^\updownarrow $ distance and then the $ D_2 $ distance.

First, we present \dslsh, a simple LSH structure for functions $ f:[0,1]\to [a,b] $ with respect to the $ L_2$ distance.
The intuition behind \dslsh is that the $ L_2 $ distance between the step functions $ f,g:[0,1]\to [a,b] $ can be approximated via a sample of $ f $ and $ g $ at the evenly spaced set of points $ \left\{i/n\right\}_{i=0}^n $. Specifically, by replacing each function $ f $ by the vector $ vec_n(f)= \left(\frac{1}{\sqrt{n}}f\left(\frac{0}{n}\right),\frac{1}{\sqrt{n}}f\left(\frac{1}{n}\right),\ldots,\frac{1}{\sqrt{n}}f\left(\frac{n-1}{n}\right)\right)$, one can show that for a large enough value of $ n\in \naturals $, $ L_2(f,g) $ can be approximated by $ L_2\left(vec_{n}(f)-vec_{n}(g)\right)$. We prove that for any two $ k $-step functions $ f,g:[0,1]\to [a,b] $, and for any $ r>0 $ and $ c>1 $: \textbf{(1)} if $ L_2(f,g)\leq r $ then $ L_2\left(vec_{n_{r,c}}(f),vec_{n_{r,c}}(g)\right) \leq c^{1/4}r$, and \textbf{(2)} if $ L_2(f,g)> cr $ then $ L_2\left(vec_{n_{r,c}}(f),vec_{n_{r,c}}(g)\right) > c^{3/4}r$ for a sufficiently large $n_{r,c}$ which is specified in Appendix~\ref{appndx:sec:l2distances}. Note that the bounds $A=c^{1/4}r $ and $B=c^{3/4}r $ are selected for simplicity, and other trade-offs are possible. The proof of this claim relies on the observation that $ (f-g)^2 $ is also a step function, and that $ L_2\left(vec_{n_{r,c}}(f),vec_{n_{r,c}}(g)\right)^2 $ is actually the left Riemann sum of $ (f-g)^2 $, so as $ n\to \infty $, it must approach  $\int_{0}^{1}(f(x)-g(x))^2dx= (L_2(f,g))^2 $. \DSLSH replaces data and query functions $ f $ with the vector samples $ vec_{n_{r,c}}(f) $, and holds an $ (c^{1/4}r, c^{3/4}r) $-\LSH structure for the $ n_{r,c} $-dimensional Euclidean distance (e.g., the \textit{Spherical-LSH} based structure of Andoni and Razenshteyn~\cite{andoni2015optimal}). The resulting structure has the parameter $ \rho=\frac{1}{2c-1} $.

In Appendix~\ref{sec:hashdelta2alternative}, we present an alternative structure tailored for the $ L_2 $ distance for general (not necessarily $k$-step) integrable functions $ f:[0,1]\to [a,b] $, based on a simple and efficiently computable asymmetric hash family which uses \rplsh as a building block. We note that this structure's $\rho$ values are larger than those of \dslsh for small values of $r$.

Next, we give \valsh — a structure for $ D_2^\updownarrow $. Recall that the mean-reduction (Section~\ref{subsec:d1updown}) of a function $ f $ is defined to be  $ \hat{f}(x)= f(x)-\int_{0}^{1} f(t) dt$. We show that the \textit{mean-reduction} has no approximation loss when used for reducing $ D_2^\updownarrow $ distances to $ L_2 $ distances, i.e., it holds that $D_2^{\updownarrow}(f,g)=L_2\left(\hat{f},\hat{g}\right)$ for any $ f,g $. Thus, to give an $ (r,cr) $-\LSH structure for $ D_2^\updownarrow $, \valsh simply holds a $ (r,cr) $-\LSH structure for $ L_2 $, and translates data and query functions $ f $ for $ D_2^\updownarrow $ to data and query functions $ \hat{f} $ for $ L_2 $.

Finally, we employ the same cloning and sliding method as in Section~\ref{subsec:d1}, to obtain an $ (r,cr) $-\LSH structure for $ D_2 $ using a structure for $ D_2^\updownarrow $.

\section{Polygon distance}\label{sec:polygondistance}
In this section (which appears in detail in Appendix~\ref{appndx:sec:polygondistance}) we consider polygons, and give efficient structures to find similar polygons to an input polygon.
All the results of this section depend on a fixed value $ m\in \naturals $, which is an upper bound on the number of vertices in all the polygons which the structure supports (both data and query polygons).
Recall that the distance functions between two polygons $ P $ and $ Q $ which we consider, are defined to be variations of the $ L_p $ distance between the turning functions $ t_P $ and $ t_Q $ of the polygons, for $ p= 1,2 $. To construct efficient structures for similar polygon retrieval, we apply the structures from previous sections to the turning functions of the polygons.\\
To apply these structures and analyze their performance, it is necessary to bound the range of the turning functions, and represent them as $ k $-step functions. Since the turning functions are $ (m+1) $-step functions, it therefore remains to compute bounds for the range of the turning function $ t_P $.\\
A coarse bound of $[-(m+1)\pi, (m+3)\pi] $ can be derived by noticing that the initial value of the turning function is in $ [0,2\pi] $, that any two consecutive steps in the turning function differ by an angle less than $ \pi $, and that the turning function has at most $ m+1 $ steps.\\
We give an improved and tight bound for the range of the turning function, which relies on the fact that turning functions may wind up and accumulate large angles, but they must almost completely unwind towards the end of the polygon traversal, such that $ t_P(1)\in [t_P(0)+\pi,t_P(0)+3\pi ]$. Our result is as follows.
\begin{theorem}[Simplified]\label{thm:turnfunctionboundsimplified}
	Let $ P $ be a polygon with $ m $ vertices. Then for the turning function $ t_P $, $\forall x\in[0,1], -\left(\floor{m/2}-1\right)\pi\leq t_P(x)\leq  \left(\floor{m/2}+3\right)\pi$, and this bound is tight.
\end{theorem}

We denote the lower and upper bounds on the range by $a_m= -\left(\floor{m/2}-1\right)\pi $ and $ b_m=\left(\floor{m/2}+3\right)\pi $ respectively,
and define $ \lambda_m $ to be the size of this range, $ \lambda_m=(2\cdot \floor{m/2}+2)\pi$.\\
Having the results above, we get LSH structures for the different corresponding polygonal distances which support polygons with at most $ m $ vertices, by simply replacing each data and query polygon by its turning function.

Regarding the distances $ D_1^\updownarrow $ and $ D_1 $, we can improve the bound above using the crucial observation that even though the range of the turning function may be of size near $ m\pi $, its span can actually only be of size approximately $ \frac{m}{2}\cdot \pi  $ (Theorem~\ref{thm:turningfunctionspanboundsimplified}), where we define the span of a function $ \phi $ over the domain $ [0,1] $, to be $ span(\phi)=\max_{x\in [0,1]}(\phi(x))-\min_{x\in [0,1]}(\phi(x)) $.

A simplified version of this result is as follows.
\begin{theorem}[Simplified]\label{thm:turningfunctionspanboundsimplified}
	Let $ Q $ be a polygon with $ m $ vertices. Then for the turning function $ t_Q $, it holds that $span(t_Q)\leq \left(\floor{m/2}+1\right)\pi=\lambda_m/2$.
	Moreover, for any $ \eps>0 $ there exists such a polygon with span at least $\left(\floor{m/2}+1\right)\pi-\eps$.
\end{theorem}

Since the $ D^{\updownarrow}_1 $ distance is invariant to vertical shifts, we can improve the overall performance of our $ D_1^\updownarrow $ LSH structure by simply mapping each data and query polygon $ P\in S $ to its vertically shifted turning function $ x\to t_P(x)-\min_{z\in [0,1]} t_P(z) $ (such that its minimal value becomes 0).  This shift morphs the ranges of the set of functions $F  $ to be contained in $ \left[0,\max_{f\in F}\left(span(f)\right)\right] $. By Theorem~\ref{thm:turningfunctionspanboundsimplified}, we can therefore use the adjusted bounds of $ a=0 $ and $ b=\lambda_m/2 $ (each function $ f\in S_0 $ is obviously non-negative, but also bounded above by $ \lambda_m/2 $ by Theorem~\ref{thm:turningfunctionspanboundsimplified}), and effectively halve the size of the range from $ \lambda_m=b_m-a_m $ to $ \lambda_m/2 $.

To summarize our results for polygons, we use the $ \tilde{O} $ notation to hide multiplicative constants which are small powers (e.g., $5$) of $ m $, $\frac{1}{r}$, and $ \frac{1}{\sqrt{c}-1} $:\\
For the $ D_1 $ distance, for any $ c>2$ we give an $ (r,cr) $-LSH structure which for $ r\ll \frac{2\lambda_m}{c} $ roughly requires $ \tilde{O}(n^{1+\rho}) $ preprocessing time and space, and $ \tilde{O}(n^{1+\rho} \log n) $ query time, where $ \rho$ is roughly $ \frac{2}{c}$. Also for $ D_1 $, for any $ c>1$ we get an $ (r,cr) $-LSH structure which for $ r\ll \lambda_m $ roughly requires $ O((nm^2)^{1+\rho}) $ preprocessing time and space, and $ O(m^{2+2\rho}n^{\rho} \log (nm)) $ query time, where $ \rho$ is roughly $ 1/c$.

For the $ D_2 $ distance, we give an $ (r,cr) $-LSH structure which requires $ \tilde{O}(n^{1+\rho}) $ preprocessing time, $ \tilde{O}(n^{1+\rho}) $ space, and $ \tilde{O}(n^{\rho}) $ query time, where $ \rho = \frac{1}{2\sqrt{c}-1}$.

	\section{Conclusions and directions for future work}
We present several novel LSH structures for searching nearest neighbors of functions with respect to the $L_1$ and the $L_2$ distances, and variations of these distances which are invariant to horizontal and vertical shifts.
This enables us to devise efficient similar polygon retrieval structures, by
applying our nearest neighbor data structures for functions, to the turning functions of the polygons.
For efficiently doing this, we establish interesting bounds on the range and span of the turning functions of $ m $-gons.

As part of our analysis, we proved that for any two functions $ f,g:[0,1]\to [a,b] $ such that $ D_1^{\updownarrow}(f,g)=r $, it holds that $L_1(\hat{f},\hat{g})\leq \left(2-\frac{r}{b-a}\right)\cdot r$.
This tight approximation guarantee may be of independent interest.
An interesting line for further research is to find near neighbor structures with tighter guarantees for simple and frequently occurring families of polygons such as rectangles, etc.

All the reductions we describe have some performance loss, which is reflected in the required space, preprocessing and query time. Finding optimal reduction parameters (e.g., an optimal value of $ \xi $ in Section~\ref{subsec:d1} for polygons) and finding more efficient reductions is another interesting line for further research. Finding an approximation scheme for the horizontal distance (similarly to the $ \left(2-\frac{r}{b-a}\right) $-approximation for the $ D^\updownarrow_1 $ distance which appears in Section~\ref{subsec:d1updown}) is another intriguing open question.
\bibliography{main}
\newpage
\section*{Appendix}
\appendix

We provide the missing parts from each section in the body of the paper. Appendix~\ref{appndx:sec:l1Distances} fills in the gaps from Section~\ref{sec:l1Distances} regarding the $ L_1 $-based distances, proves correctness of our structures, and proves our tight bound
on the approximation guarantee of the reduction from $ D_1^\updownarrow $ distances to $ L_1 $ distances by the mean-reduce transformation (Theorem~\ref{thm:meanFunctionSubtractDistanceStrongerBoundWeak}). Appendix~\ref{appndx:sec:l2distances} gives the missing parts from Section~\ref{sec:l2distances} regarding the $ L_2 $-based distances, the correctness of our structures, and proves that $ L_2 $ distances can be approximately reduced to euclidean distances via function sampling at the evenly spaced set of points $ \{i/n\}_{i=0}^n $. Appendix~\ref{appndx:sec:polygondistance} gives the missing parts from Section~\ref{sec:polygondistance} regarding the Polygon distances - it proves tight bounds on the range and the span of polygons with at most $ m $ vertices, and the correctness of the structures that build upon these bounds.

\section{Missing parts from Section~\ref{sec:l1Distances}}\label{appndx:sec:l1Distances}


\subsection{Missing parts from Subsection~\ref{subsec:d1updown}}\label{appndx:subsec:d1updown}
The following theorem gives a simple bound for the  $ L_1 $ distance between mean-reduced functions in terms of their original vertical translation-invariant $ L_1 $ distance $ D_1^{\updownarrow} $. Its proof has a similar flavor to the proof of Lemma 3 in Chen et al.~\cite{chen2013oja} for the Oja depth.

\begin{theorem}\label{thm:meanFunctionSubtractDistanceBound}
	For any two functions $ f,g:[0,1]\to [a,b] $, it holds that
	\[D_1^{\updownarrow}(f,g) \leq L_1(\hat{f},\hat{g})\leq 2\cdot D_1^{\updownarrow}(f,g).\]
\end{theorem}
\begin{proof}[Proof of Theorem~\ref{thm:meanFunctionSubtractDistanceBound}]
	We first prove the left inequality and then prove the right inequality.

	\textbf{Left inequality.} By the definition of $ D_1^{\updownarrow}(f,g) $, we have that
	\[ \int_{0}^{1} \abs{f(x)+(\bar{g}-\bar{f})-g(x)} dx\geq D_1^{\updownarrow}(f,g),\] so
	\[L_1(\hat{f},\hat{g})=\int_{0}^{1}\abs{\hat{f}(x)-\hat{g}(x)} dx= \int_{0}^{1} \abs{f(x)+(\bar{g}-\bar{f})-g(x)} dx\geq D_1^{\updownarrow}(f,g).\]

	\textbf{Right inequality.} Consider the (optimal) $\alpha \in \reals$ for which \[D_1^{\updownarrow}(f,g)= \int_{0}^{1} \abs{f(x)+\alpha-g(x)} dx .\]
	We have that
	\begin{equation}
	D_1^{\updownarrow}(f,g)=\int_{0}^{1} \abs{f(x)+\alpha -g(x)} dx  \geq \abs{\int_{0}^{1}(f(x)+\alpha -g(x)) dx}= \abs{\bar{f}+\alpha-\bar{g}}.\label{eq:D1falphag}
	\end{equation}
	Hence, for any $ x\in[0,1] $, we get that
	\begin{align*}
	\abs{\hat{f}(x)-\hat{g}(x)}
	&=\abs{\left(f(x)-\bar{f}\right)-\left(g(x)-\bar{g}\right)}
	=\abs{\left(f(x)+\alpha -g(x)\right)+\left(\bar{g}-\alpha-\bar{f}\right)}\\
	&\leq \abs{f(x)+\alpha-g(x)}+\abs{\bar{f}+\alpha-\bar{g}}
	\stackrel{(\ref{eq:D1falphag})}{\leq} D_1^{\updownarrow}(f,g)+\abs{f(x)+\alpha-g(x)},
	\end{align*}
	where the first inequality follows by the triangle inequality, and by negating the argument of the second absolute value.
	We therefore conclude that
	\[
	L_1(\hat{f},\hat{g})
	=\int_{0}^{1}\abs{\hat{f}(x)-\hat{g}(x)} dx
	\leq D_1^{\updownarrow}(f,g)+\int_{0}^{1}\abs{f(x)+\alpha-g(x)} dx
	=2\cdot D_1^{\updownarrow}(f,g).\qedhere
	\]
\end{proof}

The following proof of Theorem~\ref{thm:meanFunctionSubtractDistanceStrongerBoundWeak} gives an improved and tight bound on the ratio between $ L_1(\hat{f},\hat{g}) $ and $ D^{\updownarrow}_1(f,g) $ that depends on (decreases with) $ D^{\updownarrow}_1(f,g) $.

\begin{proof}[Proof of Theorem~\ref{thm:meanFunctionSubtractDistanceStrongerBoundWeak}]

Let $ f,g:[0,1]\to [a,b] $ be a pair of step functions for which $D_1^{\updownarrow}(f,g)=r$, let $ h $ be the step function $ h(x)=f(x)-g(x)$, let $ \bar{h}=\int_{0}^{1}h(x)dx =\bar{f}-\bar{g}$, and let $ m_h $ be an optimal vertical shift of $ h $, i.e., $ m_h = \arg \min_{\alpha\in \reals} \int_{0}^{1}\abs{h(x)-\alpha}dx $.

We observe that
\begin{equation}
L_1(\hat{f},\hat{g})=\int_{0}^{1}\abs{\hat{f}(x)-\hat{g}(x)} dx= \int_{0}^{1} \abs{f(x)+(\bar{g}-\bar{f})-g(x)} dx. \label{eq:l1hats}
\end{equation}

We first prove the left inequality and then prove the right inequality.

\textbf{Left inequality.} As in the proof of Theorem~\ref{thm:meanFunctionSubtractDistanceBound}, by the definition of $ D_1^{\updownarrow}(f,g) $ and Equation~(\ref{eq:l1hats}), we have that $ L_1(\hat{f},\hat{g})\geq D_1^{\updownarrow}(f,g) =r$.

\textbf{Right inequality.} We assume w.l.o.g.\ that $ m_h \leq \bar{h}$ (since otherwise we flip the symmetric roles of $ f $ and $ g $, so $ h $ becomes $ -h $ and $ m_h $ becomes $ -m_h $, and therefore $ m_h \leq \bar{h}$). By Equation~(\ref{eq:l1hats}) and since $ \bar{h}=\bar{f}-\bar{g} $, we get that
\begin{align*}
L_1(\hat{f},\hat{g})=\int_{0}^{1}\abs{h(x)-\bar{h}}dx
&= \int_{x\mid h(x)< m_h}\left(\bar{h}-h(x)\right)dx+ \int_{x\mid m_h\leq h(x)\leq \bar{h} }\left(\bar{h}-h(x)\right)dx \\
&+\int_{x\mid h(x)> \bar{h}}\left(h(x)-\bar{h}\right)dx.
\end{align*}

Let $ w_A=Length(\{x\mid h(x)<m_h\})\geq 0 $ be the total length of the intervals over which $ h $ is smaller than $ m_h $, and $ A= m_h-1/w_A\cdot\int_{x\mid h(x)< m_h}h(x)dx\geq 0 $ capture how smaller the mean value of $ h $ is than $ m_h $ in these intervals. Similarly, let
$ w_B=Length(\{x\mid m_h\leq h(x)\leq \bar{h}\})\geq 0 $ be the total length of the intervals over which $ h $ is between $ m_h $ and $ \bar{h} $, and $ B= 1/w_B\cdot\int_{x\mid m_h\leq h(x)\leq \bar{h}}h(x)dx-m_h\geq 0 $ capture how larger the mean value of $ h $ is than $ m_h $ in these intervals.
Finally, let
$ w_C=Length(\{x\mid h(x)>\bar{h}\})\geq 0 $ be the total length of the intervals over which $ h $ is larger than $ \bar{h} $, and $ C= 1/w_C\cdot\int_{x\mid h(x)>\bar{h}}h(x)dx-m_h\geq 0 $ captures how larger the mean value of $ h $ is than $ m_h $ in these intervals. Figure~\ref{fig:abcwawbwc} illustrates these variables. If $w_A=0 $ we define $A=0  $, if $ w_b=0 $ we define $ B=0$ (or $ \bar{h} -m_h$) and if $ w_C=0 $ we define $C=0 $.

\begin{figure}[ht]
	\centering
	\includegraphics[width=0.6\linewidth]{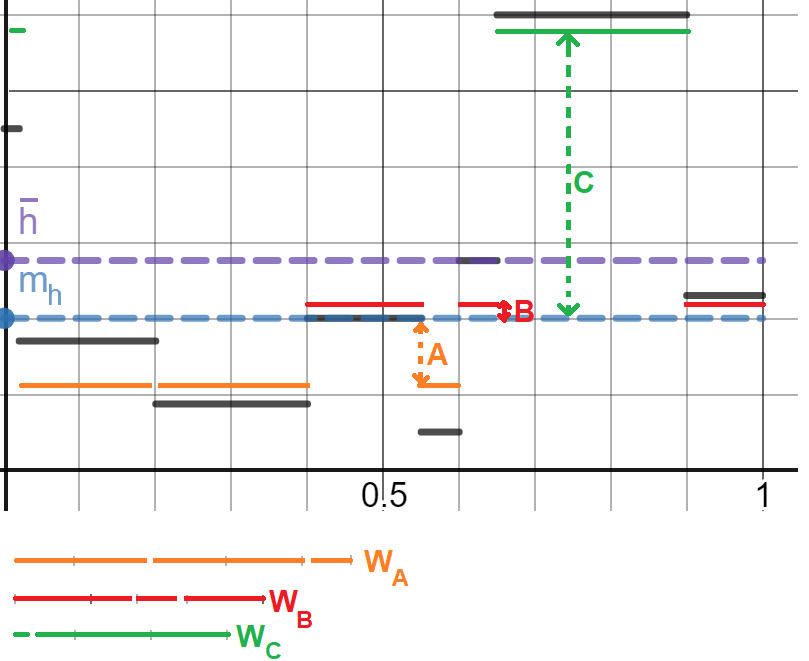}
	\caption{The function $ h $ is in black, $ m_h $ is in blue, $ \bar{h} $ is in purple, and the three widths $ w_A $, $ w_B $ and $ w_C $ and corresponding values $ A $, $ B $ and $ C $ are in orange, red and green correspondingly.}
	\label{fig:abcwawbwc}
\end{figure}

We make the following observations:
\begin{enumerate}
	\item It holds that\begin{equation}
	w_A+w_B+w_C=1. \label{eqn:sumwawbwcis1}
	\end{equation}
	\item It holds that \begin{align}
	r&=D_1^{\updownarrow}(f,g)=\min_{\alpha\in \reals} \int_{0}^{1}\abs{f(x)-g(x)-\alpha}dx=\min_{\alpha\in \reals} \int_{0}^{1}\abs{h(x)-\alpha}dx\nonumber\\
	&=\int_{0}^{1}\abs{h(x)-m_h}dx
	=\int_{x\mid h(x)< m_h}\left(m_h-h(x)\right)dx
	+ \int_{x\mid m_h\leq h(x)\leq \bar{h} }\left(h(x)-m_h\right)dx\nonumber\\
	&+\int_{x\mid h(x)> \bar{h}}\left(h(x)-m_h\right)dx=Aw_A+Bw_B+Cw_C\label{eqn:rscomputation}
	\end{align}
	\item It holds that
	\begin{equation}
	w_A\leq \frac{1}{2}~and~ w_C\leq \frac{1}{2}.\label{eq:wawcatmost0.5}
	\end{equation}
	The first claim follows since otherwise the sum of interval lengths of which for which $ h(x)<m_h $ is strictly larger than $ \frac{1}{2} $ — a contradiction to the optimality of $ m_h $, since $\int_{0}^{1}\abs{h(x)-m_h-\eps}dx<\int_{0}^{1}\abs{h(x)-m_h}dx$ for a sufficiently small $ \eps>0 $ (since most the function $ h $ is below $ m_h $). The second claim follows by a symmetric argument.
	\item Let $ M=2(b-a) $. We get that \begin{equation} C+A=1/w_C\cdot\int_{x\mid h(x)>\bar{h}}h(x)dx-1/w_A\cdot\int_{x\mid h(x)< m_h}h(x)dx\leq\max_{x}h(x)-\min_{x}h(x)\leq M, \label{eqn:bpluscatmost1}\end{equation}
	Where the first equality follows by the definitions of $ A $ and $ C $, and the last inequality follows since $ f,g:[0,1]\to [a,b] $ and therefore
	\[\forall x,y, \abs{h(x)-h(y)}\leq 2\max_{x}\abs{h(x)}=2\max_{x} \abs{f(x)-g(x)}\leq 2(b-a)=M.\]
	\item It holds that
	\begin{align}
	\bar{h}-m_h&=
	\int_{x\mid h(x)< m_h}h(x)dx+ \int_{x\mid m_h\leq h(x)\leq \bar{h} }h(x)dx +\int_{x\mid h(x)> \bar{h}}h(x)dx-m_h\nonumber\\
	&=(m_h-A)w_A+(B+m_h)w_B+(C+m_h)w_C-m_h\nonumber\\
	&=-Aw_A+Bw_B+Cw_C+(w_B+w_C+w_A-1)m_h\nonumber\\
	&\underbrace{=}_{(\ref{eqn:sumwawbwcis1})}-Aw_A+Bw_B+Cw_C\underbrace{=}_{(\ref{eqn:rscomputation})}r-2Aw_A,\label{eqn:hbar-mh}
	\end{align}
	where the second equality follows by the definitions of $ A,~B,~C,~w_A,~w_B$ and $w_C $.
\end{enumerate}

We further expand the value of $ L_1(\hat{f},\hat{g}) $:
\begin{align}
L_1(\hat{f},\hat{g})
&= \int_{x\mid h(x)< m_h}\left(\bar{h}-h(x)\right)dx+ \int_{x\mid m_h\leq h(x)\leq \bar{h} }\left(\bar{h}-h(x)\right)dx\nonumber\\
&+\int_{x\mid h(x)> \bar{h}}\left(h(x)-\bar{h}\right)dx\nonumber\\
&=(\bar{h}-m_h+A)w_A+(\bar{h}-m_h-B)w_B+(m_h+C-\bar{h})w_C\nonumber\\
&=([\bar{h}-m_h]+A)w_A+([\bar{h}-m_h]-B)w_B+(C-[\bar{h}-m_h])w_C\nonumber\\
&=Aw_A-Bw_B+Cw_C+[\bar{h}-m_h](w_A+w_B-w_C)\nonumber\\
&\underbrace{=}_{(\ref{eqn:sumwawbwcis1}),(\ref{eqn:rscomputation})}r-2Bw_B+[\bar{h}-m_h](1-2w_C )\underbrace{=}_{(\ref{eqn:hbar-mh})}r-2Bw_B+(r-2Aw_A)(1-2w_C )\nonumber\\
&=-2Aw_A (1-2w_C)-2Bw_B-2rw_C+2r\nonumber\\
&=-2Aw_A(1-2w_C)-2Bw_B +2r(r/M-w_C)+2r(1-r/M),\label{eq:l1hatfhatg}
\end{align}
where the second step follows by the definitions of $ A,~B,~C,~w_A,~w_B$ and $w_C $. \\
In order to bound the value of $ L_1(\hat{f},\hat{g}) $ from Equation~(\ref{eq:l1hatfhatg}), we observe that
\begin{equation}
r(r/M-w_C)\leq Aw_A (1-2w_C)+Bw_B.\label{eqn:r(r/M-w_C)-bound}
\end{equation}
Indeed, we split to two cases, and show that (\ref{eqn:r(r/M-w_C)-bound}) holds in each case:
\begin{itemize}
	\item If $ r\leq Mw_C $, then $ r(r/M-w_C)\leq 0\leq Aw_A (1-2w_C)+Bw_B $, where the last inequality follows by Equation~(\ref{eq:wawcatmost0.5}) and since $ A,B,w_A,w_B\geq 0 $.
	\item Otherwise, $r(r/M-w_C)\leq M(r/M-w_C)=r-Mw_c\underbrace{=}_{(\ref{eqn:rscomputation})}Aw_A+Bw_B+Cw_C-Mw_C=Aw_A (1-2w_C )+Bw_B+(C-M+2Aw_A ) w_C\leq Aw_A (1-2w_C)+Bw_B+(C-M+A) w_C\leq Aw_A (1-2w_C)+Bw_B$, where the first inequality follows by since it is given that $ r\leq b-a\leq 2(b-a)=M $ and since $ r/M-w_c>0 $, the second inequality follows by Equation~(\ref{eq:wawcatmost0.5}), and the third inequality follows by Equation~(\ref{eqn:bpluscatmost1}).
\end{itemize}
Hence, $ L_1(\hat{f},\hat{g}) \underbrace{\leq}_{(\ref{eq:l1hatfhatg}),(\ref{eqn:r(r/M-w_C)-bound})}2r(1-r/M)=\left(2-\frac{r}{b-a}\right)\cdot r$. This concludes the proof of the right inequality.

To show that the bound is tight, we define the two functions $ f_0(x)$ to be equal $a$ for $ x\in [0,1-\frac{r}{2(b-a)}] $ and $ b $ otherwise, and define the function $ g_0(x)$ to be equal $b$ for $ x\in [0,1-\frac{r}{2(b-a)}] $ and $ a $ otherwise. These functions are illustrated in Figure~\ref{fig:f0g0}.
\begin{figure}[ht]
	\centering
	\includegraphics[width=\linewidth]{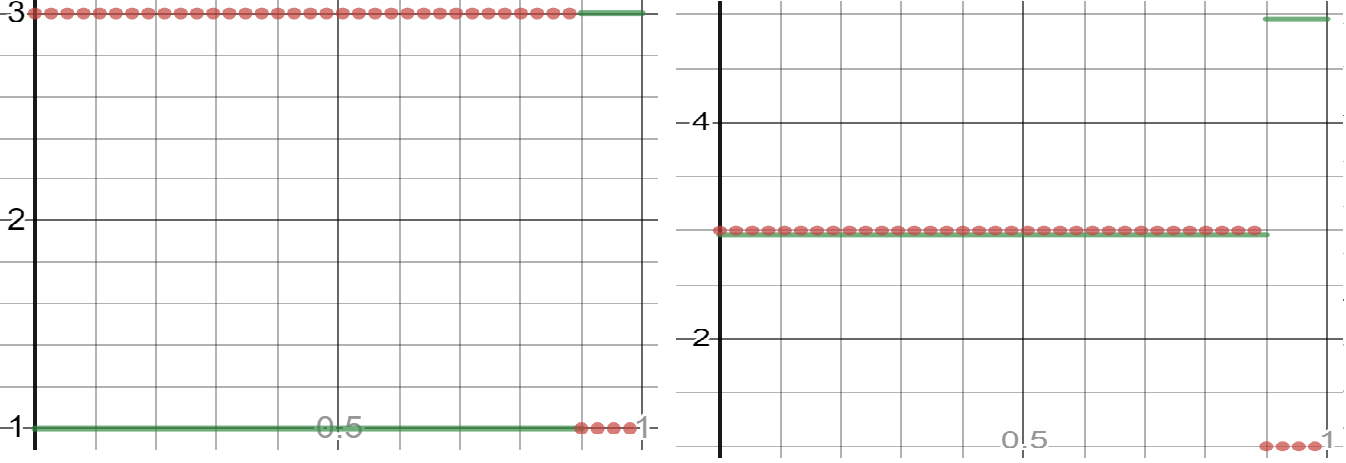}
	\caption{Left: The functions $ f_0 $ in solid green and $ g_0 $ in dotted red. Right: The optimal vertical alignment: $ f_0+b-a$ in solid green and $ g_0 $ in dotted red. Here $ a=1 $, $ b=3 $ and $ r=0.4 $, and note that $1-\frac{r}{2(b-a)}=1-\frac{0.4}{2\cdot 2}=0.9 $.}
	\label{fig:f0g0}
\end{figure}
Both these functions trivially have images contained in $ [a,b] $. The value of $ \alpha $ that minimizes $ L_1(f_0+\alpha, g_0) $ is $b-a $ (see Figure~\ref{fig:f0g0}), since it vertically aligns the first step of both functions, which is of width strictly larger than $ \frac{1}{2} $ since $r< b-a$. Since the function $ f_0+b-a-g_0 $ is equal $0$ for $ x\in [0,1-\frac{r}{2(b-a)}] $ and $ 2(b-a) $ otherwise, we conclude that $ D_1^\updownarrow(f_0,g_0)=  L_1(f_0+b-a, g_0)=0\cdot (1-\frac{r}{2(b-a)})+2(b-a)\cdot \frac{r}{2(b-a)}=r $ as required.
On the other hand, $ \bar{f_0}=a+r/2 $, so the function $ \hat{f_0}(x) $ is equal $ -r/2 $ for $ x\in [0,1-\frac{r}{2(b-a)}] $ and $ b-a-r/2 $ otherwise. Similarly, $ \bar{g_0}=b-r/2 $, so the function $ \hat{g_0}(x) $ is equal $ r/2 $ for $ x\in [0,1-\frac{r}{2(b-a)}] $ and $ a-b+r/2 $ otherwise. Hence, $ L_1(\hat{f_0},\hat{g_0})=2r/2\cdot (1-\frac{r}{2(b-a)})+2(b-a-r/2)\cdot \frac{r}{2(b-a)}= r-\frac{r^2}{2(b-a)}-\frac{r^2}{2(b-a)}=\left(2-\frac{r}{b-a}\right)\cdot r$. This concludes the proof.
\end{proof}

\begin{proof}[Proof of Theorem~\ref{thm:H1updownisLsh}]
	By the definition of $  H^{\updownarrow}_1(a,b) $, for any pair of functions $ f,g:[0,1]\to [a,b] $ we have that
	\begin{align*}
	P_{h \sim H^{\updownarrow}_1(a,b)}(h(f)=h(g))
	&=P_{h\sim H_1(a-b,b-a)}(h(\hat{f})=h(\hat{g}))
	=1-\frac{L_1(\hat{f},\hat{g})}{2(b-a)},
	\end{align*}
	where the second equality follows by Theorem~\ref{thm:L1hashCollisionProb}, noticing that $ f-g:[0,1]\to [a-b,b-a]$.

	Observe that:
	\begin{itemize}
		\item If $ D_1^{\updownarrow}(f,g)\leq r $, then by Theorem~\ref{thm:meanFunctionSubtractDistanceStrongerBoundWeak} we have that $  L_1(\hat{f},\hat{g}) \leq (2-\frac{r}{b-a})r $, so \[ P_{h \sim H^{\updownarrow}_1(a,b)}(h(f)=h(g))=1-\frac{L_1(\hat{f},\hat{g})}{2(b-a)} \geq 1-\frac{(2-\frac{r}{b-a})r}{2(b-a)}=1-\left(2-\frac{r}{b-a}\right)\cdot \frac{r}{2(b-a)}.\]
		\item If $ D_1^{\updownarrow}(f,g)\geq cr $, then since $ \hat{f} $ and $ \hat{g} $ are specific vertical shifts of $ f $ and $ g $, it follows that $  L_1(\hat{f},\hat{g}) \geq  D_1^{\updownarrow}(f,g)\geq cr $, so \[ P_{h \sim H^{\updownarrow}_1(a,b)}(h(f)=h(g)) =1-\frac{L_1(\hat{f},\hat{g})}{2(b-a)} \leq 1-c\cdot \frac{r}{2(b-a)} .\qedhere\]
	\end{itemize}
\end{proof}

\begin{proof}[Proof of Corollary~\ref{cor:lshstructD1updownverticalclones}]
	We construct the $ (r,cr) $-\LSH for the $ D_1^\updownarrow $ distance, which contains the underlying $ (r,cr) $-\rplsh structure tuned with $ a'=a-b $, $ b'=b-a $.

	We prove the correctness of our structure. Indeed, let a query function $ g:[0,1]\to [a,b] $ and a data function $ f:[0,1]\to [a,b] $.

	First we assume that $ D_{1}^{\updownarrow}(f,g)\leq r $, and prove that our structure returns (with constant probability) a function of $D_{1}^{\updownarrow}$ distance at most $ cr $ to $g $. Recall that the optimal vertical alignment $ u $ of $ f $ and $ g $ is such that a step $ f $ is vertically aligned with a step of $ g $, so there exists a step of $ f $ of height $ h_f $ and a step of $ g$ of height $ h_g $ such that $ L_1(f+h_g-h_f,g)=r $. Specifically, we have saved the clone $ f-h_f $ in the structure constructed during the preprocessing phase, and we perform a query with $ g-h_g $, so for this query, the $ L_1 $ should retrieve (with constant probability) a shifted function of $ L_1 $ distance at most $ cr $ to $g-h_g$. By the definition of the $ D_1^\updownarrow $ distance, which is invariant to vertical shifts, this returned function is of $ D_1^\updownarrow $ distance at most $ cr $ to $ g $.

	Second, we prove that no function $ f' $ for which $ D_{1}^{\updownarrow}(f,g)> cr $ is returned. Indeed, since $ D_{1}^{\updownarrow}(f,g)> cr $, then for step height $ h_f $ of $ f $ and  $ h_g $ of $ g$, it holds that $ L_1(f-h_f,g-h_g)>cr $. Therefore, for each vertical clone of $f $ in the $ L_1 $ structure, and each vertical clone of $ g $ which we query, their $L_1$ distance is strictly larger than $ cr $, and therefore $ g $ will never be returned.

	To analyze the efficiency of this structure, recall (Corollary~\ref{cor:delta1lshstruct}) that \rplsh requires $ O(n^{1+\rho}) $ space and preprocessing time, and $ O(n^\rho \log(n)) $ query time, where $ \rho=\frac{\log \left(1-\frac{r}{b-a}\right)}{\log \left(1-\frac{cr}{b-a}\right)} $. Recall again that we use \rplsh tuned with the parameters $ a'=a-b $, $ b'=b-a $, and with $ k $ copies of the data, and $ k $ queries to this structure, so we effectively have $ n' = k n$ and $ b'-a'=2(b-a) $. Therefore, by Corollary~\ref{cor:delta1lshstruct}, our structure requires $ O((nk)^{1+\rho})$ extra space and preprocessing time, and $  O(k^{1+\rho}n^\rho \log(nk))  $ query time, where
	$ \rho=\log \left(1-\frac{r}{2(b-a)}\right)/\log \left(1-\frac{cr}{2(b-a)}\right) $.
\end{proof}

\subsection{Missing parts from Subsection~\ref{subsec:d1}}\label{appndx:subsec:d1}

\begin{proof}[Proof of Theorem~\ref{thm:DdistanceReduction}]
	This proof is analogous to the correctness proof from Corollary~\ref{cor:lshstructD1updownverticalclones}, but with respect to slides, discontinuity points and the $ D_1$ distance rather than vertical shifts, step heights and the $ D_1^\updownarrow $ distance respectively. This proof relies on the fact that the optimal vertical alignment between a pair of step functions $ f $ and $ g $ is such that a discontinuity point of $ f $ is aligned with a discontinuity point of $ g $.
\end{proof}

\begin{proof}[Proof of Corollary~\ref{cor:lshstructD1}]
	We construct the $ (r,cr) $-\sclsh for the $ D_1 $ distance, which contains the underlying $ (r,cr) $-\mrlsh structure tuned with $ a'=a $, $ b'=b+2\pi $.

	To analyze the efficiency of this structure, recall (Corollary~\ref{cor:lshstructD1updown}) that \mrlsh requires $ O(n^{1+\rho}) $ space and preprocessing time, and $ O(n^\rho \log(n)) $ query time, where $\rho =\log \left(1-(2-2\tilde{r})\cdot \tilde{r}\right)/\log \left(1-c\tilde{r}\right)$ and $ \tilde{r}=r/(2(b-a)) $. Recall again that \sclsh uses \mrlsh tuned with the parameters $ a'=a $ and $ b'=b+2\pi $ (so $ b'-a'=\omega$), and with $ (k+1) $ copies of the data, and $ (k+1) $ queries to this structure, so we effectively have $ n' = (k+1) n$. Therefore, by Corollary~\ref{cor:lshstructD1updown}, \sclsh requires $ O\left((n(k+1))^{1+\rho}\right)= O\left((nk)^{1+\rho}\right) $ space and preprocessing time, and $ O\left((k+1)^{1+\rho}n^\rho \log(n)\right) =O\left(k^{1+\rho}n^\rho \log(n)\right)$ query time, with $\rho =\log \left(1-(2-2\tilde{r})\cdot \tilde{r}\right)/\log \left(1-c\tilde{r}\right)$ and $ \tilde{r}=r/(2\omega) $.
\end{proof}

\begin{proof}[Proof of Corollary~\ref{cor:lshstructD1withsslsh}]
	We construct the $ (r,cr) $-\sclsh for the $ D_1 $ distance, which contains the underlying $ (r,cr) $-\sslsh structure tuned with $ a'=a $, $ b'=b+2\pi $.
	The rest of the proof is similar to the proof of Corollary~\ref{cor:lshstructD1}, but respect to \sslsh instead of \mrlsh.
\end{proof}

\section{Detailed presentation of \texorpdfstring{$ L_2 $-}{}based distances (Section~\ref{sec:l2distances})}\label{appndx:sec:l2distances}
In this section we give a detailed explanation regarding the structure for the $ L_2 $, $ D_2^\updownarrow $ and $ D_2$ distances.

\subsection{Structure for \texorpdfstring{$ L_2 $}{}}\label{sec:hashdelta2}
In this section, we present \dslsh, a simple LSH structure for functions $ f:[0,1]\to [a,b] $ with respect to the  $ L_2$ distance.
The intuition behind \dslsh is that any step function $ f:[0,1]\to [a,b] $ can be approximated arbitrarily well by a step function with steps over the domains $ \left\{\left[(i-1)/n,i/n\right]\right\}_{i=1}^{n} $, and the $ L_2 $ distance between two such functions is closely related to the $ \ell_2 $ distance between the vectors of step heights of the approximations.

To formalize this intuition, we introduce the notion of a left Riemann sum as follows.
Let a function $ \phi:[0,1]\to \reals $, and let $ P=\{[x_0,x_1],\ldots,[x_{n-1},x_n]\} $ be a partition of $ [0,1] $, where $ a=x_0<x_1<\ldots<x_n=b $. The \textit{left Riemann sum} $ S $ of $ \phi $ over $ [0,1] $ with the partition $ P $ is defined to be $ S= \sum_{i=1}^{n}\phi(x_{i-1})(x_i-x_{i-1})$.
It holds that for any step function $ \phi $, as the maximum size of a partition element shrinks to zero, the left Riemann sums converge to the integral of $ \phi $ over $ [0,1] $.

Let $ f,g:[0,1]\to [a,b] $ be a pair of step functions. We sample $ f $ and $ g $ at $ n $ equally spaced points to create the vectors $ vec_n(f)$ and $vec_n(g) $ respectively, where for a function $ \phi:[0,1]\to [a,b] $ and an integer $ n\in \naturals $, we define $ vec_n(\phi) $ to be $\left(\frac{1}{\sqrt{n}}\phi\left(\frac{0}{n}\right),\frac{1}{\sqrt{n}}\phi\left(\frac{1}{n}\right),\ldots,\frac{1}{\sqrt{n}}\phi\left(\frac{n-1}{n}\right)\right)$.
It is easy to see that $L_2\left(vec_n(f),vec_n(g)\right)^2 $ is exactly the \textit{left Riemann sum} of the function $ (f-g)^2 $ with respect to the partition $ P_n=\left\{\left[\frac{0}{n},\frac{1}{n}\right],\left[\frac{1}{n},\frac{2}{n}\right],\ldots,\left[\frac{n-1}{n},\frac{n}{n}\right]\right\} $ of $ [0,1] $. Thus, the $ L_2 $ distance between $ f $ and $ g $ can be approximated arbitrarily well, with a sufficiently large $ n $ via
\[L_2(f,g)=\sqrt{\int_{0}^{1}(f(x)-g(x))^2 dx}\approx
\sqrt{L_2\left(vec_n(f),vec_n(g)\right)^2}= L_2\left(vec_n(f),vec_n(g)\right).\footnote{For practical applications, the quantization parameter $ n $ enables a trade-off between precision and efficiency, where a small value of $ n $ results in a quicker but less accurate scheme (i.e., return false positives), and a large value of $ n $ results in a slower but more accurate scheme.}\]

Given the parameters $ r>0 $ and $ c>1 $, \dslsh expects to receive an additional global value $ n_{r,c} $ as an input, satisfying that for any $ n\geq n_{r,c} $ the approximation above holds between any query function $ g $ and input function $ f $, in the sense that
\begin{enumerate}[(i)]
	\item \label{i1}If $ L_2(f,g)\leq r $ then $ L_2\left(vec_n(f),vec_n(g)\right) \leq c^{1/4}r$, and
	\item \label{i2}If $ L_2(f,g)> cr $ then $ L_2\left(vec_n(f),vec_n(g)\right) > c^{3/4}r$.\footnote{The bounds $A=c^{1/4}r $ and $B=c^{3/4}r $ are arbitrarily selected, such that they satisfy $ r<A<B<cr $. This selection gives rise to a reduction with efficiency which depends on the parameters $ c'=\sqrt{c} $ and  $ n_{r,c} $, and ultimately gives rise to the parameter $ \rho=\frac{1}{2c-1} $ in Corollary~\ref{cor:delta2struct}. The formula for the most efficient selection of $ A $ and $ B $ is omitted since it is not elegant, but in practice one would use the optimized values. A similar arbitrary selection of $ c'=\sqrt{c} $, which could be optimized, is made in Corollary~\ref{cor:D2qlshstruct}.}
\end{enumerate}

We give a simple global value $ n_{r,c} $ in terms of the range $ [a,b] $, the number of steps $ k $ and the parameters $ r,c $.

\begin{theorem}\label{thm:stepdelta2bound}
	Let $ r>0 $ and $ c>1 $, let $ f,g:[0,1]\to [a,b] $ be two $ k $-step functions, and let $n_{r,c}=\frac{2k(b-a)^2}{(\sqrt{c}-1)r^2} $. Then, for any $ n\geq n_{r,c} $, (\ref{i1}) and (\ref{i2}) are satisfied.\footnote{Recall that the domain of a $ k $-step function can be split into $ k $ intervals such that $ f $ constant in each interval.}
\end{theorem}
\begin{proof}[Proof of Theorem~\ref{thm:stepdelta2bound}]
	Fix $ n\geq n_{r,c} $, and define $ \Delta $ to be the absolute value difference between $ L_2\left(vec_n(f),vec_n(g)\right)^2 $ and $\int_{0}^{1}(f(x)-g(x))^2 dx $. To show that (\ref{i1}) and (\ref{i2}) are satisfied, we give an upper bound on $ \Delta $.
	We represent each of the terms above as a sum of $ n $ elements, and get that $ L_2\left(vec_n(f),vec_n(g)\right)^2 = \sum_{i=1}^{n}\frac{1}{n}\cdot \left(f\left(\frac{i-1}{n}\right)-g\left(\frac{i-1}{n}\right)\right)^2, $ and $\int_{0}^{1}(f(x)-g(x))^2 dx= \sum_{i=1}^{n}\int_{\left[\frac{i-1}{n},\frac{i}{n}\right]} (f(x)-g(x))^2 dx $. Denote the $ i $'th elements in the sums above by $ \alpha_i=\frac{1}{n}\cdot \left(f\left(\frac{i-1}{n}\right)-g\left(\frac{i-1}{n}\right)\right)^2 $ and $\beta_i=\int_{\left[\frac{i-1}{n},\frac{i}{n}\right]} (f(x)-g(x))^2 dx  $, respectively. It holds that
	\begin{equation}\label{eq:delta2vsnorm}
	\Delta=\abs{L_2\left(vec_n(f),vec_n(g)\right)^2-L_2(f,g)^2}=\abs{\sum_{i=1}^{n}(\alpha_i-\beta_i)}\leq\sum_{i=1}^{n}\abs{\alpha_i-\beta_i}.
	\end{equation}

	We bound the sum above by proving that most of the $ \alpha_i $'s are near the corresponding $ \beta_i $'s, and that the size of the set of indices $ i $ for which $ \alpha_i$ is far from $ \beta_i$ is relatively small.

	Since $ f,g:[0,1]\to [a,b] $ are $ k $-step functions, it follows that $ (f-g)^2 $ is a $2k $-step function, and has a range bounded in $[0,(b-a)^2]$.

	We split the analysis over the indices $ i $, depending on whether the interval $ I=\left[\frac{i-1}{n},\frac{i}{n}\right] $ does or does not contain a discontinuity point of $ (f-g)^2 $.

	\begin{itemize}
		\item If it does not contain such a discontinuity point, both the functions $ f $ and $ g $ are constant in the interval $ I$, and so is $ (f-g)^2 $. For the constant function $ (f-g)^2 $, any Riemann summand (specifically $ \alpha_i $) is exactly the integral ($ \beta_i $), and therefore $ \alpha_i=\beta_i $, i.e., $ \abs{\alpha_i-\beta_i}=0 $.

		\item If it does contain such a discontinuity point, since the range of $(f-g)^2$ is bounded in $[0,(b-a)^2]$, it holds that $ \alpha_i,\beta_i\in [0,\frac{1}{n}\cdot (b-a)^2] $ and therefore
		$\abs{\alpha_i-\beta_i}\leq (b-a)^2/n.$
	\end{itemize}

	Since there are at most $ 2k-1<2k $ discontinuity points of $ (f-g)^2 $, it holds that
	\begin{align}
	\Delta&\leq\sum_{i=1}^{n}\abs{\alpha_i-\beta_i}\\
	&
	=\sum_{\left\{i\mid(f-g)^2 \text{ is constant in } \left[\frac{i-1}{n},\frac{i}{n}\right]\right\}}\abs{\alpha_i-\beta_i}+\sum_{\left\{i\mid(f-g)^2 \text{ is not constant in } \left[\frac{i-1}{n},\frac{i}{n}\right]\right\}}\abs{\alpha_i-\beta_i}\nonumber\\
	&\leq 0+ 2k\cdot \frac{(b-a)^2}{n}= \frac{2k(b-a)^2}{n}\leq \frac{2k(b-a)^2}{n_{r,c}}=(\sqrt{c}-1)r^2,\label{eq:delta2vecsquaredappx}
	\end{align}
	where the first inequality follows by Equation~(\ref{eq:delta2vsnorm}), the second inequality follows by the cases above, the third inequality follows since $ n\geq n_{r,c} $, and the last equality follows by the definition of $ n_{r,c} $.

	We now prove the required facts one after the other:
	\begin{enumerate}
		\item If $L_2(f,g)\leq r $, then by Equation~(\ref{eq:delta2vecsquaredappx}) we get that
		\begin{align*}
		L_2\left(vec_n(f),vec_n(g)\right)^2 &\leq L_2(f,g)^2 +(\sqrt{c}-1)r^2
		\leq r^2+(\sqrt{c}-1)r^2=\sqrt{c}r^2,
		\end{align*}
		and therefore $ L_2\left(vec_n(f),vec_n(g)\right)\leq  c^{1/4}r$.

		\item If $L_2(f,g)> cr $, then by Equation~(\ref{eq:delta2vecsquaredappx}), and since $ c>1 $ we get that
		\begin{align*}
		L_2\left(vec_n(f),vec_n(g)\right)^2 &\geq L_2(f,g)^2 -(\sqrt{c}-1)r^2
		> (cr)^2-(\sqrt{c}-1)r^2\\
		&>c^2r^2-c^{3/2}\cdot(\sqrt{c}-1)r^2=c^{3/2}r^2,
		\end{align*}
		and therefore $L_2\left(vec_n(f),vec_n(g)\right)>c^{3/4}r $.\qedhere
	\end{enumerate}
\end{proof}

In Section~\ref{sec:D2poly}, we will indirectly use \dslsh for step functions, which are derived from turning functions of $ m $-gons. In this case, the value of $ n_{r,c} $ is derived using bounds we give over the range and span of such functions.

\DSLSH works as follows.

\paragraph*{Preprocessing phase}
Given the parameters $r > 0$ and $c > 1$ and the corresponding parameter $ n_{r,c} $, we transform each function $ f $ to $ vec_{n_{r,c}}(f)$, and store the resulted vectors in an $ (r',c'r') $-\LSH structure for the $ n_{r,c} $-dimensional Euclidean distance (e.g., the Spherical-LSH based structure of Andoni and Razenshteyn~\cite{andoni2015optimal}), tuned with the parameters $r'=c^{1/4}r $ and $ c'=\sqrt{c}$.

\paragraph*{Query phase}
Let $ g $ be a query function. We query the $ (r,cr) $-\LSH structure for the Euclidean distance constructed in the preprocessing phase with the query $ vec_{n_{r,c}}(g)$.

We now prove that \dslsh is an $ (r,cr) $-\LSH structure.

\begin{corollary}\label{cor:delta2struct}
	For any $ r>0 $ and $ c>1 $, \dslsh is an $ (r,cr) $-\LSH structure for the  $ L_2 $ distance. \DSLSH requires $ O(n^{1+\rho}+n_{r,c}\cdot n)$ space, $ O(n_{r,c}\cdot n^{1+\rho})$ preprocessing time, and $ O(n_{r,c}\cdot n^{\rho}) $ query time, where $ \rho=\frac{1}{2c-1} $ and $ n $ is the size of the data set.\footnote{Note that we do not necessarily need to store the vectors $ vec_{n_{r,c}}(f)$, but rather only the original functions $ f $ and the hashes of each $ vec_{n_{r,c}}(f)$, keeping with it a pointer back directly to its original function $ f $. This allows us to remove the term $ n_{r,c}\cdot n $ (which represents the space required to store the data itself) from the space requirements of \dslsh.}
\end{corollary}
\begin{proof}[Proof of Corollary~\ref{cor:delta2struct}]
	We first show that \dslsh is an $ (r,cr) $-\LSH structure, and then analyze its performance.

	Recall that \dslsh relies on an $(r',c'r')$-\LSH structure for the Euclidean distance. Thus, in order to prove that \dslsh is an $ (r,cr) $-\LSH structure for the  $ L_2 $ distance, we show that:
	\begin{enumerate}
		\item $ c'>1 $,
		\item for any input function $ f $ and query function $ g$ such that $L_2(f,g)\leq r$, it holds that\\ $~L_2\left(vec_{n_{r,c}}(f),vec_{n_{r,c}}(g)\right)\leq r',$ and
		\item for any input function $ f $ and query function $ g$ such that $L_2(f,g)> cr$, it holds that\\ $~L_2\left(vec_{n_{r,c}}(f),vec_{n_{r,c}}(g)\right)> c'r'$.
	\end{enumerate}
	The proofs of these facts are as follows.
	\begin{enumerate}
		\item $ c'=\sqrt{c}>1$, since $ c>1 $.
		\item Assume that $L_2(f,g)\leq r$. We prove that $L_2\left(vec_{n_{r,c}}(f),vec_{n_{r,c}}(g)\right)\leq r'$. Indeed, by the definition of $ n_{r,c} $,
		\[
		L_2\left(vec_{n_{r,c}}(f),vec_{n_{r,c}}(g)\right)\leq c^{1/4}r= r'.
		\]
		\item Assume that $L_2(Q,x)>cr$. We prove that $L_2\left(vec_{n_{r,c}}(f),vec_{n_{r,c}}(g)\right)> c'r' $. Indeed, by the definition of $ n_{r,c} $,
		\[
		L_2\left(vec_{n_{r,c}}(f),vec_{n_{r,c}}(g)\right)>c^{3/4}r=\sqrt{c}\cdot c^{1/4}r= c'r'.
		\]
	\end{enumerate}

	To analyze the time and space bounds, recall that the data structure of Andoni and Razenshteyn~\cite{andoni2015optimal} has $ O(d\cdot n^{\rho }) $ query time, requires $ O(n^{1+\rho}+dn) $ space and $ O(d\cdot n^{1+\rho}) $ preprocessing time, where $\rho=\frac{1}{2c^2-1} $ and $ d $ is the dimension of the euclidean space. By the definition of \dslsh, we use the structure of Andoni and Razenshteyn for $ n $ points, in the dimension $ n_{r,c} $, and with an approximation ratio (the LSH parameter $ c $) of $ \sqrt{c} $. Hence, the query time is $ O(n_{r,c}\cdot n^{\rho }) $, the space is $ O(n^{1+\rho}+n_{r,c}\cdot n) $ and the preprocessing time is $ O(n_{r,c}\cdot n^{1+\rho})$, for $ \rho =  \frac{1}{2(\sqrt{c})^2-1}=\frac{1}{2c-1}$.
\end{proof}

\begin{corollary}\label{cor:d2alignedlshstructstep}
	For any $ r>0 $ and $ c>1 $, there is an $ (r,cr) $-\LSH structure for the $ L_2$ distance for $ n $ functions, each is $ k $-step function with ranges contained in $ [a,b] $. This structure requires $ O(n^{1+\rho}+n_{r,c}\cdot n)$ extra space, $ O(n_{r,c}\cdot n^{1+\rho}) $ preprocessing time, and $ O(n_{r,c}\cdot n^{\rho }) $ query time, where $ \rho=\frac{1}{2c-1} $ and where  $ n_{r,c}=\frac{2k(b-a)^2}{(\sqrt{c}-1)r^2} $.\footnote{Andoni and Razenshteyn~\cite{andoni2015optimal} have an additional exponent of $ o(1) $ in the efficiency terms, which arises from their assumption that the memory required to store a hash function, and time it takes to evaluate a single hash value is $ n^{o(1)} $, and that $ 1/p_1=n^{o(1)} $. In the introduction we stated that we omit these terms, so we indeed omit the additional exponent of $ o(1) $.}
\end{corollary}
\begin{proof}[Proof of Corollary~\ref{cor:d2alignedlshstructstep}]
	This is immediate by Theorem~\ref{thm:stepdelta2bound} and Corollary~\ref{cor:delta2struct}.
\end{proof}

\subsection{Structure for \texorpdfstring{$ D_2^{\updownarrow} $}{}}\label{subsec:D2updown}
In this section, we present \valsh, a simple LSH structure for $ k $-step functions $ f:[0,1]\to [a,b] $ with the vertical translation-invariant $ L_2 $ distance, $ D_2^{\updownarrow} $.
Lemma~\ref{lma:reduceMeanSolvesD2} shows how to reduce the $ D_2^{\updownarrow} $ distance to the  $ L_2 $ distance.
\begin{lemma}\label{lma:reduceMeanSolvesD2}
	For any pair of functions $ f,g:[0,1]\to \reals $, it holds that
	$D_2^{\updownarrow}(f,g)=L_2\left(\hat{f},\hat{g}\right)$.
\end{lemma}
\begin{proof}[Proof of Lemma~\ref{lma:reduceMeanSolvesD2}]
	This proof is direct from an observation from Arkin et al.~\cite{arkin1991efficiently}.
\end{proof}

It follows from Lemma~\ref{lma:reduceMeanSolvesD2} that if we shift each function $ f $ to its mean-reduction $ \hat{f} $, the $ D^{\updownarrow}_2 $ distance reduces to the  $ L_2 $ distance.

The \valsh structure works as follows.

\paragraph*{Preprocessing phase}
We are given the parameters $r > 0$ and $c > 1$, $ k \in \naturals$. We transform each data function $ f $ to $\hat{f}$, and store the transformed data functions in an $ (r,cr) $ \dslsh structure for the $ L_2 $ distance, for functions with ranges bounded in $ [a-b,b-a] $, and with the parameter $ n_{r,c} $ tuned to $ n_{r,c}=\frac{8k(b-a)^2}{(\sqrt{c}-1)r^2}  $.\footnote{For any function $ f:[0,1]\to [a,b] $, its average $ \bar{f} $ must satisfy $ \bar{f}\in [a,b] $. Thus, the range of $ \hat{f} $ is in $ [a-b,b-a] $.}$ ^, $\footnote{This value of $ n_{r,c} $ is precisely the value of $ n_{r,c} $ from Section~\ref{sec:hashdelta2}, but with respect to the range $[a-b,b-a]$. Specifically, for any $ n\geq n_{r,c} $: if $ L_2(\hat{f},\hat{g})\leq r $ then $ L_2\left(vec_n(\hat{f}),vec_n(\hat{g})\right) \leq c^{1/4}r$, and if $ L_2(\hat{f},\hat{g})> cr $ then $ L_2\left(vec_n(\hat{f}),vec_n(\hat{g})\right) > c^{3/4}r$.}

\paragraph*{Query phase}
Let $ g $ be a query function. We query the \dslsh structure constructed in the preprocessing phase with the query $ \hat{g}$.

The following is a corollary of Lemma~\ref{lma:reduceMeanSolvesD2}.
\begin{corollary}\label{cor:d2updownlshstruct}
	For any $ r>0 $ and $ c>1 $, \valsh is an $ (r,cr) $-\LSH structure for the $ D^{\updownarrow} _2$ distance for $ n $ functions, each of which is a $ k $-step function with ranges bounded in $ [a,b] $. \VALSH requires $ O(n^{1+\rho}+n_{r,c}\cdot n)$ space, $ O(n_{r,c}\cdot n^{1+\rho})$ preprocessing time, and $ O(n_{r,c}\cdot n^{\rho }) $ query time, where $ \rho=\frac{1}{2c-1} $ and $ n $ is the size of the data set and where  $ n_{r,c}=\frac{8k(b-a)^2}{(\sqrt{c}-1)r^2} $.
\end{corollary}
\begin{proof}[Proof of Corollary~\ref{cor:d2updownlshstruct}]
	This is immediate by Lemma~\ref{lma:reduceMeanSolvesD2}, and by the fact that the mean-reduced functions have ranges which are contained in $ [a-b,b-a] $.
\end{proof}

\subsection{Structure for \texorpdfstring{$ D_2 $}{}}\label{subsec:D2}
We follow the same ideas as described in Section~\ref{subsec:d1}.

\begin{theorem}\label{thm:DdistanceReductionl2}
	\SCLSH from Section~\ref{subsec:d1} but with an internal \LSH structure for the $ D_2^\updownarrow $ distance (rather than one for the  $ D_1^\updownarrow $ distance) is an $ (r,cr) $-\LSH structure for the $ D_2 $ distance.
\end{theorem}
\begin{proof}[Proof of Theorem~\ref{thm:DdistanceReductionl2}]
	This proof is identical to that of Theorem~\ref{thm:DdistanceReduction}, but with our structure for the $ D_2^\updownarrow $ distance.
\end{proof}

\begin{corollary}\label{cor:D2qlshstruct}
	For any $ r>0 $ and $ c>1 $, there is an $ (r,cr) $-\LSH structure for the $ D_2$ distance for $ n $ functions, each of which is a $ k $-step function with range bounded in $ [a,b] $. This structure requires $ O\left(\left(n(k+1)\right)^{1+\rho}+n_{r,c}\cdot n(k+1)\right)$ extra space, $ O\left(n_{r,c}\cdot \left(n(k+1)\right)^{1+\rho}\right)$ preprocessing time, and $ O\left(n_{r,c}\cdot (k+1)^{1+\rho}\cdot n^{\rho }\right) $ query time, where $ \rho=\frac{1}{2\sqrt{c}-1} $, $ n_{r,c}=\frac{8(k+1)\omega^2}{(\sqrt{c}-1)r^2} $ and $\omega =b+2\pi-a$.
\end{corollary}
\begin{proof}[Proof of Corollary~\ref{cor:D2qlshstruct}]

	We construct the $ (r,cr) $-\sclsh for the $ D_2$ distance, which as opposed to section~\ref{subsec:d1}, here it contains an underlying $ (r,cr) $-\valsh structure for the $ D_2^\updownarrow $ distance (rather than a structure for $ D_1^\updownarrow $), tuned with $ a'=a $, $ b'=b+2\pi$.

	To analyze the efficiency of this structure, we define $ \omega $ to be $b+2\pi-a$, and use Corollary~\ref{cor:d2updownlshstruct} with the parameters $ r'=r $, $ c'=c $, $ k'=k+1 $, $ b'-a'= \omega$ and with $ n'= n (k+1)$, and observe that the resulting value of $ n_{r,c} $ is $ n_{r,c}=\frac{8(k+1)\omega^2}{(\sqrt{c}-1)r^2} $. Note that the value of $ n'= n (k+1) $ is an upper bound on the number of data functions (including all clones) in the underlying \valsh structure, which is $ (k+1)n $.
\end{proof}

\subsection{Alternative structure for \texorpdfstring{$ L_2 $}{}}\label{sec:hashdelta2alternative}
In this section we present a simple asymmetric hash family for functions $ f:[0,1]\to [a,b] $ with respect to the $ L_2 $ distance.
We use it to derive an LSH structure tailored for the $ L_2 $ distance, which unlike the structure from Appendix~\ref{sec:hashdelta2}, uses  simpler and more efficient hash functions, and does not require embedding functions in high dimensional euclidean spaces. Specifically, unlike the structure from Appendix~\ref{sec:hashdelta2}, this structure can handle not only $ k $-step functions, but also general integrable functions. We note however that the $\rho$ values are larger than those from \dslsh (see Appendix~\ref{sec:hashdelta2}) for small values of $r$.

Our asymmetric hash family contains pairs of data and query hash functions
$ H_2(a,b)=\left\{\left(h^D_{(x,y_1,y_2,UseSecond)},h^Q_{(x,y_1,y_2,UseSecond)}\right)\right\}$,
where the points $x$ are uniformly selected from the segment $ [0,1]$, the points $y_1$ and $y_2$ are uniformly and independently selected from the segment $ [a,b]$, and UseSecond is uniformly selected from $\{0,1\}$.

In order to define $h^D$ and $h^Q$, we recall the $h_{(x,y)}$ hash from Section~\ref{sec:hashdelta1}, which receives a function $ f: [0,1]\to [a,b]$, and returns $ 1 $ if $ f $ is vertically above the point $(x,y) $, returns $ -1 $ if $ f $ is vertically below $ (x,y) $, and $0$ otherwise.
In our hash $H_2$, both $ h^D_{(x,y_1,y_2,UseSecond)} $ and $ h^Q_{(x,y_1,y_2,UseSecond)} $ receive a function $ f: [0,1]\to [a,b]$, and return two concatenated hash values. For the first hash value they return $h_{(x,y_1)}(f)$. For the second hash value, if $UseSecond=0$ then they both return the reserved hash value $\mbox{\enquote{$*$}}$, and otherwise they return opposite answers: $ h^D_{(x,y_1,y_2,UseSecond)} $ returns $h_{(x,y_2)}(f)$ and $ h^Q_{(x,y_1,y_2,UseSecond)} $ returns $1-h_{(x,y_2)}(f)$. Below is the formal definition of these hash functions,
\begin{align*}
    h^D_{(x,y_1,y_2,UseSecond)}(f) &=
		\left(h_{(x,y_1)}(f)~,~~~ h_{(x,y_2)}(f) ~~~~~\mbox{if } UseSecond=1\mbox{ else } \mbox{\enquote{$*$}}\right)\mbox{ and}\\
    h^Q_{(x,y_1,y_2,UseSecond)}(f) &=
		\left(h_{(x,y_1)}(f)~,~ 1-h_{(x,y_2)}(f) ~\mbox{ if } UseSecond=1\mbox{ else } \mbox{\enquote{$*$}}\right).
\end{align*}

The intuition behind this hash family is that for any two functions $ f,g:[0,1]\to [a,b]$ and fixed value $x\in [0,1]$, the collision probability of $h^D_{(x,y_1,y_2,UseSecond)}(f)$ and $h^Q_{(x,y_1,y_2,UseSecond)}(g)$ over the first hash value is $1-\frac{|f(x)-g(x)|}{b-a}$ as in Section~\ref{sec:hashdelta1}. Since the second hash values of $h^D_{(x,y_1,y_2,UseSecond)}(f)$ and $h^Q_{(x,y_1,y_2,UseSecond)}(g)$ are equal $\mbox{\enquote{$*$}}$ with probability 0.5, and otherwise are equal $h_{(x,y_2)}(f)$ and $1-h_{(x,y_2)}(g)$ respectively (which collide exactly for values of $y_2$ between $f(x)$ and $g(x)$), then the second hash collision probability is $0.5+0.5\frac{|f(x)-g(x)|}{b-a}$. Hence, the total hash collision probability for a fixed value of $x$ is $0.5-0.5\frac{|f(x)-g(x)|^2}{(b-a)^2}$, and integrating over the uniform sample of $x$ from $[0,1]$, we get that the total hash collision probability is $0.5-\frac{L_2(f,g)^2}{(b-a)^2}$, i.e., a decreasing function of $L_2(f,g)$.

\begin{theorem}\label{thm:L2hashCollisionProb}
	For any two functions $ f,g:[0,1]\to [a,b] $, we have that
	\[P_{(h^D,h^Q)\sim H_2(a,b)}(h^D(f)=h^Q(g))=0.5-\frac{L_2(f,g)^2}{(b-a)^2}.\]
\end{theorem}
\begin{proof}
	Fix $ x\in [0,1] $, and denote by $ U(S) $ the uniform distribution over a set $ S $. Recall from the proof of Theorem~\ref{thm:L1hashCollisionProb} that
	\begin{align*}
	P_{y_1\sim U([a,b])}(h_{(x,y_1)}(f)&=h_{(x,y_1)}(g))=1-\frac{\abs{f(x)-g(x)}}{b-a},
	\end{align*}
	and we similarly get that
	\begin{align*}
	P_{y_2\sim U([a,b])}(h_{(x,y_2)}(f)&=1-h_{(x,y_1)}(g))=\frac{\abs{f(x)-g(x)}}{b-a}.
	\end{align*}
	Therefore, since the first and second hash values are independent, we get that
	\begin{align*}
	&P_{y_1,y_2\sim U([a,b])}(h^D_{(x,y_1,y_2,UseSecond)}(f)=h^Q_{(x,y_1,y_2,UseSecond)}(g) \mid UseSecond=1)\\
	&=\left(1-\frac{\abs{f(x)-g(x)}}{b-a}\right)\cdot\left(\frac{\abs{f(x)-g(x)}}{b-a}\right),
	\end{align*}
	and since for $UseSecond=0$ the second hash value is constant $*$, we get that
	\begin{align*}
	&P_{y_1,y_2\sim U([a,b])}(h^D_{(x,y_1,y_2,UseSecond)}(f)=h^Q_{(x,y_1,y_2,UseSecond)}(g) \mid UseSecond=0)\\
	&=1-\frac{\abs{f(x)-g(x)}}{b-a}.
	\end{align*}
	Thus, by the law of total probability over the random variable $UseSecond$,
	\begin{align*}
	&P_{y_1,y_2\sim U([a,b]),~UseSecond\sim U(\{0,1\})}(h^D_{(x,y_1,y_2,UseSecond)}(f)=h^Q_{(x,y_1,y_2,UseSecond)}(g))\\
	&=0.5\cdot \left(1-\frac{\abs{f(x)-g(x)}}{b-a}\right)\cdot\left(\frac{\abs{f(x)-g(x)}}{b-a}\right)+0.5\cdot \left(1-\frac{\abs{f(x)-g(x)}}{b-a}\right)\\
	&= 0.5\cdot \left(1-\frac{\abs{f(x)-g(x)}}{b-a}\right)\cdot \left(1+\frac{\abs{f(x)-g(x)}}{b-a}\right)=0.5-0.5\frac{|f(x)-g(x)|^2}{(b-a)^2}.
	\end{align*}
	Hence, by the law of total probability over the random variable $x$,
	\begin{align*}
	P_{(h^D,h^Q)\sim H_2(a,b)}(h^D(f)=h^Q(g))&=\int_{0}^{1} \left(0.5-0.5\frac{|f(x)-g(x)|^2}{(b-a)^2}\right) dx\\
	&=0.5-\frac{L_2(f,g)^2}{(b-a)^2},
	\end{align*}
	where the last step follows by the linearity of the integral and by the definition of $L_2(f,g)$.
\end{proof}

\begin{corollary}\label{cor:delta2lshstruct}
		For any $ r>0 $ and $ c>1 $, one can construct an $ (r,cr) -$LSH structure for the $ L_2 $ distance for $ n $ functions with ranges bounded in $ [a,b] $. This structure requires $ O(n^{1+\rho}) $ space and preprocessing time, and has $ O(n^\rho \log(n)) $ query time, where $ \rho=\frac{\log \left(0.5-\frac{r^2}{(b-a)^2}\right)}{\log \left(0.5-\frac{c^2r^2}{(b-a)^2}\right)}$.
\end{corollary}
\begin{proof}
	Fix $ r> 0 $ and $ c>1 $. By the general result of Indyk and Motwani~\cite{indyk1998approximate}, it suffices to show that $ H_2(a,b) $ is an $ (r,cr, 0.5-\frac{r^2}{(b-a)^2},0.5-\frac{c^2r^2}{(b-a)^2}) $-\LSH for the  $ L_1 $ distance.

	Indeed, by Theorem~\ref{thm:L2hashCollisionProb}, $P_{(h^D,h^Q)\sim H_2(a,b)}(h^D(f)=h^Q(g)) =0.5-\frac{L_2(f,g)^2}{(b-a)^2} $, so we get that
	\begin{itemize}
		\item If $  L_2 (f,g)\leq r $, then $ P_{(h^D,h^Q)\sim H_2(a,b)}(h^D(f)=h^Q(g))=0.5-\frac{L_2(f,g)^2}{(b-a)^2}\geq 0.5-\frac{r^2}{(b-a)^2}.$
		\item If $  L_2 (f,g)\geq cr $, then $ P_{(h^D,h^Q)\sim H_2(a,b)}(h^D(f)=h^Q(g))=0.5-\frac{L_2(f,g)^2}{(b-a)^2}\leq 0.5-\frac{c^2r^2}{(b-a)^2}.$
	\end{itemize}
\end{proof}

We note that similar methods to those presented in Appendix~\ref{subsec:D2updown} and Appendix~\ref{subsec:D2} can be applied to the structure from Corollary~\ref{cor:delta2lshstruct} (rather than the structure from Corollary~\ref{cor:delta2struct}) in order to build structures for the $D^\updownarrow_2$ and $D_2$ distances.

\section{Detailed presentation of Polygon distance (Section~\ref{sec:polygondistance})}\label{appndx:sec:polygondistance}
In this section we consider polygons, and give efficient structures to find similar polygons to an input polygon.
All the results of this section depend on a fixed value $ m\in \naturals $, which is an upper bound on the number of vertices in all the polygons which the structure supports (both data and query polygons).
Recall that the distance functions between two polygons $ P $ and $ Q $ which we consider, are defined based on variations of the $ L_p $ distance between the turning functions $ t_P $ and $ t_Q $ of the polygons, for $ p= 1,2 $. To construct efficient structures for similar polygon retrieval, we apply the structures from the previous sections to the turning functions of the polygons.
We assume that no three consecutive vertices on the boundary of the polygon are collinear.

\subsection{Structure for the polygonal \texorpdfstring{$ D_1 $}{} distance}\label{sec:D1poly}
Our structure is constructed by applying an \LSH structure for the $ D_1 $ distance to the turning functions of the polygons. It is necessary to bound the range of the turning functions in order to construct such a structure and analyze its performance. The bounds of the turning functions depend on $ m $, which is an upper bound on the number of vertices in polygons which we support (both data and query polygons).

A coarse bound of $[-(m+1)\pi, (m+3)\pi] $ for the range of the turning function $ t_P $ can be derived by noticing that the initial value of the turning function is in $ [0,2\pi] $, that any two consecutive steps in the turning function differ by an angle less than $ \pi $, and that the turning function has at most $ m+1 $ steps.\footnote{A turn of approximately $ \pi $ corresponds to a \enquote{U Turn} in $ P $. A turn of exactly $ \pi $ cannot occur, since we assume that no three consecutive vertices are collinear.}

We give an improved and tight bound for the range of the turning function, which relies on the fact that turning functions may wind up and accumulate large angles, but they must almost completely unwind towards the end of the polygon traversal, such that $ t_P(1)\in [t_P(0)+\pi,t_P(0)+3\pi ]$.\footnote{If the reference point is selected to be in the middle of an edge of $ P $, then in fact $ t_P(1)= t_P(0)+2\pi$. The extreme values of $ t_P(1)\in \{t_P(0)+\pi,t_P(0)+3\pi \}$ can be approximately achieved by setting the reference point to be a vertex of $ P $, and by making the last \enquote{turn} be either a left or a right \enquote{U Turn}.}$ ^{,\ref{ftnt:a}} $
\begin{theorem}\label{thm:turnfunctionbound}
	Let $ P $ be a polygon with $ m $ vertices. Then for the turning function $ t_P $, it holds that \[\forall x\in[0,1], -\left(\floor{m/2}-1\right)\pi\leq t_P(x)\leq  \left(\floor{m/2}+3\right)\pi.\]
	Moreover, this bound is tight, i.e., for any $ \eps>0 $ there exist two $ m $-gons $ P,Q $ with turning functions $ t_P,t_Q $ and values $ x_P,x_Q $ respectively such that $ t_P(x_P)\geq \left(\floor{m/2}+3\right)\pi-\eps$ and $ t_Q(x_Q)\leq -\left(\floor{m/2}-1\right)\pi+\eps$.
\end{theorem}
\begin{proof}[Proof of Theorem~\ref{thm:turnfunctionbound}]

	Let $ t_1,\ldots,t_n $ be the sequence of the heights of the $ n\in \{m,m+1\} $ steps of $ t_P $ (ordered from first to last).\footnote{The number of steps $ n $ of the turning function is either $ m $ or $ m+1 $, since a turning function starting from the middle of an edge has $ m+1 $ steps, and a turning function starting from a vertex has $ m $ steps.\label{ftnt:a}} To bound the values of $ t_i $ and prove the theorem's statement, we can bound the sequence of differences of heights of consecutive steps. Therefore, for $ i=2,\ldots,m $ we define $ \Delta_i $ to be the $ i $'th difference $\Delta_i= t_i-t_{i-1} $.\footnote{We only define $ \Delta_i $ until $ m $ and not until $ n $, which is either $ m $ or $m+1 $, since we do not need $  \Delta_{m+1}$ if it exists as we handle this case separately.}
	We make two observations regarding the sequence $ \{\Delta_i\}_{i=2}^{m} $:
	\begin{enumerate}[(i)]
		\item \label{prop2} $ \forall i=2,\ldots,m,~\abs{\Delta_i}\leq \pi $, and
		\item \label{prop1} $\sum_{i=2}^{m} \Delta_i\in [\pi,3\pi]$.
	\end{enumerate}
	 The first follows since the angle between two subsequent edges is at most $ \pi $, and the second follows since the sum is equal the last step height minus the first step height, which should be either $ 2\pi $ or at most one step distance away from $ 2\pi $ (depending whether the turning function starts on a vertex or not).

Let $ Neg=\{i\in \{2,\ldots,m\}\mid \Delta_i \leq 0\} $ be the set of indices $ i $ for which $ \Delta_i \leq 0$, and let $ Pos=\{i\in \{2,\ldots,m\}\mid \Delta_i>0\} $ be the set of indices $ i $ for which $ \Delta_i >0$ and let $ s=\floor{m/2}  $.

We first prove the theorem's left inequality and then the right inequality:
\begin{enumerate}
	\item
	We assume by contradiction that there exists a $ k\in \{1,\ldots,n\} $ for which $ t_k<-(s-1)\pi$.

	We split into two cases. In the first case, $ k\leq m $, and in the second $ k>m $, which implies that $ n=m+1 $ and $ k=n $.

	In the first case, since $ t_1\geq 0 $, we get that
	$\sum_{i\in Neg} \Delta_i\leq \sum_{i=1}^{k}\Delta_i =t_k-t_1<  -(s-1)\pi-t_1\leq -(s-1)\pi$, so we apply (\ref{prop1}) to conclude that $
	\sum_{i\in Pos} \Delta_i =  \sum_{i=2}^{m} \Delta_i - \sum_{i\in Neg} \Delta_i > \pi +(s-1)\pi =s\pi$.
	By applying (\ref{prop2}) to both these equations it follows that $ \abs{Neg}\geq s $ and $ \abs{Pos}\geq s+1 $, so $ \abs{Neg}+\abs{Pos}\geq 2s+1\geq m $, in contradiction to the fact that $ Neg $ and $ Pos $ are two disjoint subsets of $ \{2,\ldots,m\} $.

	In the second case, it must be that $ t_n=t_{m+1}=t_1+2\pi>0 $, so the left inequality holds.

	\item 	Assume by contradiction that there exists an $  k\in \{1,\ldots,n\}$ for which $ t_k>(s+3)\pi$.

	We split into two cases. In the first case, $ k\leq m $, and in the second $ k>m $, which implies that $ n=m+1 $ and $ k=n $.

	In the first case, since $ t_1\leq 2\pi $, we get that
	$\sum_{i\in Pos} \Delta_i\geq \sum_{i=1}^{m}\Delta_i =t_m-t_1>  (s+3)\pi-t_1\geq (s+1)\pi$, so we apply (\ref{prop1}) to conclude that $
	\sum_{i\in Neg} \Delta_i = \sum_{i=2}^{m} \Delta_i - \sum_{i\in Pos} \Delta_i< 3\pi-(s+1)\pi=-(s-2)\pi$.
	By applying (\ref{prop2}) to both these equations it follows that $ \abs{Pos}\geq s+2 $ and $ \abs{Neg}\geq s-1 $, so $ \abs{Neg}+\abs{Pos}\geq 2s+1\geq m $, in contradiction to the fact that $ Neg $ and $ Pos $ are two disjoint subsets of $ \{2,\ldots,m\} $.

	In the second case, it must be that $ t_n=t_{m+1}=t_1+2\pi<4\pi \leq (s+3)\pi $ for any $ s\geq 1 $ (obviously, $ m\geq 3 $ so $ s\geq 1 $), so the right inequality holds.
\end{enumerate}

	We now describe the polygon $ P $ for which the turning function $ t_P(x) $ admits a value of\\$ \left(\floor{m/2}+3\right)\pi-\eps$, and then describe a polygon $ Q $ for which the turning function $ t_Q(x) $ admits a value of $ -\left(\floor{m/2}-1\right)\pi+\eps$.

	We assume that $ m $ is an even number $ m=2k $, and handle the case where $ m $ is odd separately.

	We build $ P $ via the following process. We consider a polyline oriented from left to right with edges of length $ 1,1+\eps,\ldots,1+(k-1)\eps,1+(k-1)\eps,\ldots,1+\eps,1 $, such that the segment $ i $ and the segment $ 2k-i $ have the same length. We consider the natural order over the vertices (points), and define $ A $ to be the leftmost point, $ G $ to be the rightmost point and $ F $ to be the left neighbor of $ G $. This is illustrated in Figure~\ref{fig:foldpart1}.
	\begin{figure}[ht]
		\centering
		\includegraphics[width=0.95\linewidth]{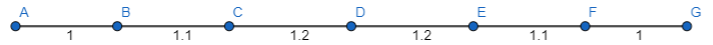}
		\caption{Here $ m=6 $, $ k= 3$ and $ \eps=0.1 $. $ A $ is the leftmost point, $ G $ is the rightmost point and $ F $ is the left neighbor of $ G $.}
		\label{fig:foldpart1}
	\end{figure}

	Next, fold the right half of the polyline over the left half such that the vertices $ A $ and $ G $ of the polyline connect, and $ F $ becomes the clockwise neighbor of $ A $. This is illustrated in Figure~\ref{fig:foldpart2}.
	\begin{figure}[ht]
		\centering
		\includegraphics[width=0.95\linewidth]{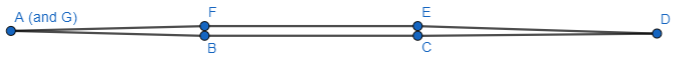}
		\caption{Fold of the right half over the left half, and connection of $ A $ and $ G $.}
		\label{fig:foldpart2}
	\end{figure}

	Then, we start folding (rolling) up and inwards the polygon segments (in a clockwise fashion), such that in each step we wrap the rolled part of the polygon around an additional pair of segment (see Figure~\ref{fig:foldparts}-(a),(b)). Next, we rotate the tightened polygon it such that the first edge in the counter-clockwise traversal (the edge $ FA $ in Figure~\ref{fig:foldparts}) has an initial turning function value of $ 2\pi-\frac{\eps}{2} $ (see Figure~\ref{fig:foldparts}-(c)). Finally, we tighten the fold such that all the edges create an angle which is $\ll \frac{\eps}{2} $ with each other, and such that the orientation of $ FA $ does not change (see Figure~\ref{fig:foldparts}-(d)). We define $ P $ to be this polygon, and its reference point to be $ F $.
	\begin{figure}[ht]
		\centering
		\subcaptionbox{First fold}{{\includegraphics[width=0.45\linewidth]{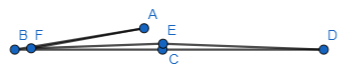} }}
		\qquad
		\subcaptionbox{Second (last) fold}{{\includegraphics[width=0.45\linewidth]{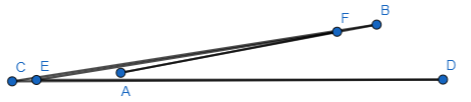} }}
		\qquad
		\subcaptionbox{Rotated such that over the first edge ($ FA $) the turning function is equal $ 2\pi- \frac{\eps}{2}$}{{\includegraphics[width=0.45\linewidth]{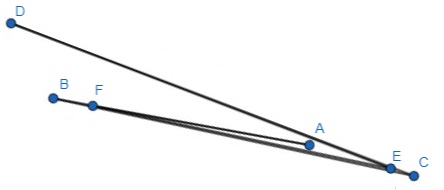} }}
		\qquad
		\subcaptionbox{Tightening fold}{{\includegraphics[width=0.45\linewidth]{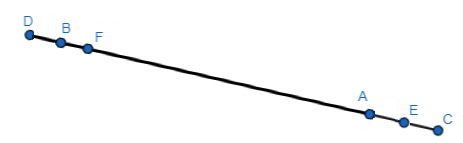} }}
		\caption{The folding and orientation process of $ P $.}
		\label{fig:foldparts}
	\end{figure}

	We now show that $ t_P(x) $ admits a value of $ \left(k+3\right)\pi-\eps$. Indeed, the initial angle of the turning function is $ 2\pi-\frac{\eps}{2} $, and in each of the first $ k+1 $ breakpoints of $ t_P $ ($A, B,C $ and $ D $ in the figures above) the turning function grows by approximately $ \pi$. Since we have tightened the polygon $ P $, each turning function angle is of absolute value difference which is $\ll \frac{\eps}{2} $ from $ 2\pi-\frac{\eps}{2}+\pi s $ for some $ s\in \naturals $. It is therefore easy to see that the angle of the turning function after the $ (k+1) $'th breakpoint is of absolute value difference at most $ \frac{\eps}{2} $ from $ 2\pi-\frac{\eps}{2}+(k+1)\pi$, and is therefore at least $ \left(k+3\right)\pi-\eps=\left(\floor{\frac{m}{2}}+3\right)\pi-\eps$.

	By using symmetric arguments, we can show that the polygon $ Q $ for which the turning function $ t_Q(x) $ admits a value of $ -\left(\floor{m/2}-1\right)\pi+\eps$ is simply the reflection of $ P $ with respect to the $ y $-axis, with the same reference point $ A $.

	We finally address the case where $ m $ is odd. In this case we take the polygons $ P $ and $ Q $ from above for the
	even number $m-1 $, and add an arbitrary vertex in the middle of an arbitrary edge of $ P $ and $ Q $ respectively. This does not
	affect the turning function, and $ t_P $ admits a value of $ \left(\floor{\frac{m-1}{2}}+3\right)\pi-\eps=\left(\floor{\frac{m}{2}}+3\right)\pi-\eps$, and $ t_Q $ admits a
	value of $-\left(\floor{\frac{m-1}{2}}-1\right)\pi+\eps=-\left(\floor{\frac{m}{2}}-1\right)\pi+\eps$.
\end{proof}

By Theorem~\ref{thm:turnfunctionbound}, it follows that all turning functions must have a range bounded between $a_m= -\left(\floor{m/2}-1\right)\pi $ and $ b_m=\left(\floor{m/2}+3\right)\pi $.
We define $ \lambda_m $ to be size of the range in which the turning functions reside. That is $ \lambda_m=b_m-a_m =(2\cdot \floor{m/2}+2)\pi$.

Let $ r>0 $ and $c>1 $, where $ m $ is an upper bound on the number of vertices in the data and query polygons. We give an $ (r,cr)- $LSH structure for the polygonal $ D_1 $ distance, which works as follows.
In the preprocessing phase, we store the turning function $ t_P $ of all the polygons $ P\in S $ in the $ (r,cr)- $structure for the $ D_1$ distance guaranteed by Corollary~\ref{cor:lshstructD1}, with the parameters $ a=a_m $, $ b=b_m $ and $ k=m+1 $.
Given a query polygon $ Q $, we query the structure from the preprocessing phase with $ t_Q $. Using Theorem~\ref{thm:turnfunctionbound}, Corollary~\ref{cor:lshstructD1} and the fact that the turning functions are $ (m+1) $-step functions with ranges bounded in $ [a_m,b_m] $, one can show that the structure above requires $ O((nm)^{1+\rho})$ extra space and preprocessing time, and $  O(m^{1+\rho}n^\rho \log(nm)  $ query time, where
$\rho =\log \left(1-(2-2\tilde{r})\cdot \tilde{r}\right)/\log \left(1-c\tilde{r}\right)$ and $ \tilde{r}=r/(2\lambda_m+4\pi) $.

We improve the performance of this structure by the following crucial observations. The first is that the performance of both our \LSH structures for the $ D_1 $ distance depend on the size of the range $[a,b] $ of the set of functions $ f:[0,1]\to [a,b] $ which it supports (the smaller the range size, the better). The second is that even though the range of the turning function of an $ m $-gon may be of size near $ m\pi $, its span can actually only be of size approximately $ \frac{m}{2}\cdot \pi  $ (Theorem~\ref{thm:turningfunctionspanbound}), where we defined the span of a function $ \phi $ over the domain $ [0,1] $, to be $ span(\phi)=\max_{x\in [0,1]}(\phi(x))-\min_{x\in [0,1]}(\phi(x)) $.
Since the $ D_1$ distance is invariant to vertical shifts, rather than mapping each data and query polygon $ P$ directly to is turning function, we map it to its vertically shifted turning function $ x\to t_P(x)-\min_{z\in [0,1]} t_P(z) $, effectively morphing the range to be $ [0,\lambda_m/2] $ which is half the size of the original range.

\begin{theorem}\label{thm:turningfunctionspanbound}
	Let $ Q $ be a polygon with $ m $ vertices. Then for the turning function $ t_Q $, it holds that $span(t_Q)\leq \left(\floor{m/2}+1\right)\pi=\lambda_m/2.$
	Moreover, this bound is tight, i.e., for any $ \eps>0 $ there exists an $ m $-gon $ P $ with turning function $ t_P$ such that $span(t_Q)\geq \left(\floor{m/2}+1\right)\pi-\eps$.
\end{theorem}
\begin{proof}[Proof of Theorem~\ref{thm:turningfunctionspanbound}]
	Similarly to the proof Theorem~\ref{thm:turnfunctionbound}, we assume that $ t_1,\ldots,t_n $ (for $ n\in \{m,m+1\} $) are the sequence of the heights of the steps of $ t_Q $ (ordered from first to last). For $ i=2,\ldots,m $ we define $ \Delta_i $ to be the $ i $'th difference $\Delta_i= t_i-t_{i-1} $, and we let $ N=\{i\in \{2,\ldots,m\}\mid \Delta_i<0\} $ be the set of indices $ i $ for which $ \Delta_i \leq 0$, let $ P=\{i\in \{2,\ldots,m\}\mid \Delta_i>0\} $ be the set of indices $ i $ for which $ \Delta_i >0$.

	Additionally, let $ t_i $ and $ t_j $ be the step heights for which $ span(t_Q)=\abs{t_i-t_j} $, and assume w.l.o.g.\ that $ i>j $. We define $ s=\floor{m/2} $ (therefore $ m\leq 2s+1 $), and we show that $ \abs{t_i-t_j}\leq (s+1)\pi$.
	We split into two cases. In the first case, $ i\leq m $, and in the second $ i>m $, for which it must be that $ n=m+1 $ and $ i=n $.

	In the case where $ i\leq m $, we have that
	\begin{align*}
	\abs{t_i-t_j}&=\abs{\Sigma_{k\in \{j+1,\ldots,i\}}\Delta_k}=\abs{\Sigma_{k\in P\cap\{j+1,\ldots,i\}}\Delta_k+\Sigma_{k\in N\cap\{j+1,\ldots,i\}}\Delta_k}\\
	&=\abs{\Sigma_{k\in P\cap\{j+1,\ldots,i\}}\Delta_k-\Sigma_{k\in N\cap\{j+1,\ldots,i\}}\abs{\Delta_k}}\\
	&\leq \max\left(\Sigma_{k\in P\cap\{j+1,\ldots,i\}}\Delta_k,\Sigma_{k\in N\cap\{j+1,\ldots,i\}}\abs{\Delta_k}\right)\\
	&\leq \max\left(\Sigma_{k\in P}\Delta_k, \Sigma_{k\in N}\abs{\Delta_k}\right)=
	\max\left(\Sigma_{k\in P}\Delta_k, -\Sigma_{k\in N}\Delta_k\right)=\max\left(S_P,-S_N\right),
	\end{align*}
	where the third equality follows by the definition of $ N $, the first inequality follows since $ \abs{\gamma -\phi}\leq \max(\gamma,\phi) $ for any $ \gamma,\phi\geq 0 $, and the last equality follows by defining $ S_P $ and $ S_N $ to be $ \Sigma_{i\in P}\Delta_i $ and $ \Sigma_{i\in N}\Delta_i $ respectively.

	By the proof of Theorem~\ref{thm:turnfunctionbound}, we get that $ S_P+S_N=\sum_{i=2}^{m} \Delta_i \in  [\pi,3\pi] $. It follows that $ -S_N\leq S_P-\pi $, so $\max\left(S_P, -S_N\right)= S_P$ and therefore $ \abs{t_i-t_j}\leq S_P $.

    To conclude the required bound, it therefore suffices to prove that $ S_P\leq (s+1)\pi $. Indeed, we assume by contradiction that $ S_P> (s+1)\pi $.
    Since $ \forall i,\abs{\Delta_i}\leq \pi$, by the definitions of $ S_P $ and $ S_N $ it follows that $ \abs{P}\geq \frac{S_P}{\pi} $ and $ \abs{N}\geq \frac{-S_N}{\pi} $. Therefore $ \abs{P}\geq  s+2$ and so $ \abs{N}=(m-1)-\abs{P}\leq (m-1)-(s+2) \leq (2s+1-1)-(s+2)=s-2$, and therefore $ S_N\geq -\abs{N}\pi\geq -(s-2)\pi =(2-s)\pi$. We get that $ S_P+S_N> (s+1)\pi +(2-s)\pi=3\pi$. This contradicts the fact that $ S_P+S_N\in  [\pi,3\pi] $.

    In the other case where $ n=m+1 $ and $ i=m+1 $, we define $ \Delta_{m+1} =t_{m+1}-t_{m}$ and extend $ N $ and $ P $ to include this index as appropriate. We now have that $ S_P +S_N=2\pi$, and $ \abs{P}+\abs{N}=m $. As before, we bound $ S_P $ from above, by assuming by contradiction that $ S_P> (s+1)\pi $.
    As before, $ \abs{P}\geq \frac{S_P}{\pi} $ and $ \abs{N}\geq \frac{-S_N}{\pi} $. Therefore $ \abs{P}\geq  s+2$ and so $ \abs{N}=m-\abs{P}\leq m-(s+2) \leq (2s+1)-(s+2)=s-1$, and therefore $ S_N\geq -\abs{N}\pi\geq -(s-1)\pi =(1-s)\pi$. We get that $ S_P+S_N> (s+1)\pi +(1-s)\pi=2\pi$, in contradiction to the fact that $ S_P +S_N=2\pi$ .

	It remains to prove that the bound is tight. Indeed, we use the same polygon $ P $ from the proof Theorem~\ref{thm:turnfunctionbound}, which has a point $ x_P $ for which $ t_P(x_P)\geq \left(\floor{m/2}+3\right)\pi-\eps$, and that it satisfies $ t_P(0)= 2\pi-\eps/2 $. Hence, the span of the turning function $ t_P(x) $ is at least $ \left(\floor{m/2}+3\right)\pi-\eps-(2\pi-\eps/2)= \left(\floor{m/2}+1\right)\pi-\eps/2>\left(\floor{m/2}+1\right)\pi-\eps$.

\end{proof}

The improved structure described above, is identical to the previous one however with a range of $ [a,b] $ where $ a=0 $ and $ b=\lambda_m/2 $. It has the following guarantees:
\begin{theorem}\label{thm:polygonD1struct}
	For any $ r>0 $ and $ c>2-\frac{r}{\lambda_m/2+2\pi} $, there exists an $ (r,cr) $-LSH structure for the polygonal $ D_1 $ distance for $ n $ $ m $-gons. This structure requires $ O((nm)^{1+\rho})$ extra space and preprocessing time, and $  O(m^{1+\rho}n^\rho \log(nm))  $ query time,\\ where
	$\rho =\log \left(1-(2-2\tilde{r})\cdot \tilde{r}\right)/\log \left(1-c\tilde{r}\right)$ and $ \tilde{r}=r/(\lambda_m+4\pi) $
\end{theorem}
\begin{proof}[Proof of Theorem~\ref{thm:polygonD1struct}]
	We use the underlying \LSH structure for the $ D_1$ distance from Corollary~\ref{cor:lshstructD1} with the vertically shifted turning functions of our polygons.

	First, observe that vertical shifts do not change the $ D_1^\updownarrow $ and $ D_1 $ distances, and since the span of the turning functions is at most $ \lambda_m/2 $, then the vertically shifted turning functions are bounded in $ [0,\lambda_m/2] $.
	Second, observe that $ c>2-\frac{r}{\lambda_m/2+2\pi}= 2-\frac{r}{b+2\pi-a}$, where the first inequality follows by our constraint on $ c$, and the last inequality follows since $ b-a=\lambda_m/2 $.

	Therefore, the theorem's statement follows by applying Corollary~\ref{cor:lshstructD1}, by the definition of the $ D_1 $ distance, by substituting in the values of $ a=0 $, $ b=\lambda_m/2 $, $ r $, $ c $ and $ k=m+1 $, and since $b-a=b_m-a_m=\lambda_m/2$. We use $ k=m+1 $ since turning functions of polygons are $(m+1)$-step functions.
\end{proof}

\begin{theorem}\label{thm:polygonD1structother}
	For any $ r>0 $ and $ c>1 $, there exists an $ (r,cr) $-LSH structure for the polygonal $ D_1 $ distance for $ n $ $ m $-gons. This structure requires $ O((nm^2)^{1+\rho})$ extra space and preprocessing time, and $  O(m^{2+2\rho}n^\rho \log(nm))  $ query time, where
	\[\rho=\log \left(1-\frac{r}{\lambda_m+4\pi}\right)/\log \left(1-\frac{cr}{\lambda_m+4\pi}\right) .\]
\end{theorem}
\begin{proof}[Proof of Theorem~\ref{thm:polygonD1structother}]
	We use the underlying \LSH structure for the $ D_1 $ distance guaranteed by Corollary~\ref{cor:lshstructD1withsslsh}.
	The proof follows from similar arguments to those in the proof of Theorem~\ref{thm:polygonD1struct}, but applying Corollary~\ref{cor:lshstructD1withsslsh} (rather than Corollary~\ref{cor:lshstructD1}) with the following parameters $ r'=r $, $ c'=c $, $  a=0$, $ b=\lambda_m/2 $ and $ k=m+1 $.
\end{proof}

\subsection{Structure for the polygonal \texorpdfstring{$ D_2 $}{} distance}\label{sec:D2poly}
We give an LSH structure for the polygonal $ D_2 $ distance over $ m $-gons, which works as follows.
In the preprocessing phase, given a data set $ S $ of polygons, $ r>0 $ and $c>1 $, similarly to Section~\ref{sec:D1poly}, we store the vertically shifted turning function $ x\to t_P(x)-\min_{z\in [0,1]} t_P(z) $ of all the polygons $ P\in S $ in an $ (r,cr)- $structure for the $ D_2 $ distance guaranteed by Corollary~\ref{cor:D2qlshstruct} with the morphed range of $ a=0 $, $ b=\lambda_m/2 $ and $ k=m+1 $. The functions $ a_m $, $ b_m $ and $ \lambda_m=(2\cdot \floor{m/2}+2)\pi$ are defined in Section~\ref{sec:D1poly}. Given a query polygon $ Q $, we query the structure from the preprocessing phase with the vertically adjusted version of $ t_Q $.

Since our new range is of size $ \lambda_m/2 $, we get an improved structure with the following performance guarantees.

\begin{theorem}\label{thm:polygonD2struct}
	The structure described above is an $ (r,cr) $-\LSH structure for the polygonal $ D _2$ distance for $ n $ $ m $-gons.\\
	 This structure requires $ O\left(\left(n(m+2)\right)^{1+\rho}+n_{r,c}\cdot n(m+2)\right)$ extra space,\\ $ O\left(n_{r,c}\cdot \left(n(m+2)\right)^{1+\rho}\right)$ preprocessing time, and $ O\left(n_{r,c}\cdot (m+2)^{1+\rho}\cdot n^{\rho }\right) $ query time, where $ \rho=\frac{1}{2\sqrt{c}-1} $, $ n_{r,c}=\frac{8(m+2)\omega^2}{(\sqrt{c}-1)r^2} $ and $\omega =\lambda_m/2+2\pi $.
\end{theorem}
\begin{proof}[Proof of Theorem~\ref{thm:polygonD2struct}]
	Since our structure is identical to that from Theorem~\ref{thm:polygonD1structother}, but using an internal structure for $ D_2 $ (from Corollary~\ref{cor:D2qlshstruct}) rather than a structure for $ D_1 $, the proofs are the same except the fact that this proof uses Corollary~\ref{cor:D2qlshstruct} instead of Corollary~\ref{cor:lshstructD1withsslsh}.
\end{proof}

\end{document}